\documentclass[a4paper, 12pt]{article}
\usepackage[right=0.5in,left=0.5in]{geometry}
\usepackage{relsize}
\usepackage{amsthm}

\newtheorem{proposition}{Proposition}[section]
\newtheorem{lemma}[proposition]{Lemma}

\newtheorem{corollary}[proposition]{Corollary}
\theoremstyle{remark}
\newtheorem{remark}[proposition]{Remark}
\theoremstyle{definition}
\newtheorem{definition}[proposition]{Definition}

\usepackage[bf,justification=centering,font=small]{caption}
\usepackage{lscape}
\usepackage{color}
\usepackage{graphicx}
\usepackage{amssymb}
\usepackage{amsthm}
\usepackage{lineno}
\usepackage{mathtools}
\usepackage{appendix}
\usepackage{float}
\usepackage{hyperref}
\usepackage{amsmath}
\usepackage{multirow}
\usepackage[utf8]{inputenc}

\title{Epidemic reconstruction in a phylogenetics framework: transmission trees as partitions}
\author{Matthew Hall
\and
Andrew Rambaut}
\begin{document}
\maketitle

\section{Abstract}

The reconstruction of transmission trees for epidemics from genetic data has been the subject of some recent interest. It has been demonstrated that the transmission tree structure can be investigated by augmenting internal nodes of a phylogenetic tree constructed using pathogen sequences from the epidemic with information about the host that held the corresponding lineage. In this paper, we note that this augmentation is equivalent to a correspondence between transmission trees and partitions of the phylogenetic tree into connected subtrees each containing one tip, and provide a framework for Markov Chain Monte Carlo inference of phylogenies that are partitioned in this way, giving a new method to co-estimate both trees. The procedure is integrated in the existing phylogenetic inference package BEAST.

\section{Introduction}

The increasing availability of faster and cheaper sequencing technologies is making it possible to acquire genetic data on the pathogens involved in outbreaks and epidemics at a very fine resolution. It is likely that in future outbreaks where most or all infected hosts can be identified, one or more pathogen nucleotide sequences will be available from each one as a matter of course. Identification of a high proportion of hosts is plausible in several scenarios, such as agricultural outbreaks, where the infected unit will usually be taken to be the farm and considerable government resources will be employed to identify every one, HIV, where almost all infected individuals will eventually seek treatment, and epidemics involving a population that can be closely monitored, such as those occurring in hospitals or prisons. As a result, much recent work has been performed to develop computational methods to analyse data of this kind, combining it with more traditional epidemiological data \cite{cottam_integrating_2008, aldrin_modelling_2011, ypma_unravelling_2011, jombart_reconstructing_2011, morelli_bayesian_2012, ypma_relating_2013, jombart_bayesian_2014, didelot_bayesian_2014, mollentze_bayesian_2014}. A Bayesian Markov Chain Monte Carlo (MCMC) approach is almost always employed, as the probability spaces involved are of very high dimension and mathematically complicated; the only exception is the study by Aldrin et al. \cite{aldrin_modelling_2011}, which used a maximum-likelihood method.

The most frequent approach to this problem has been to attach a mutation model to a model of transmission, making simplifications that link the process of nucleotide substitution to host-to-host transmission events. Commonly, transmission events are assumed to coincide with times of most recent common ancestor of isolates, ignoring any within-host diversity; the assumption being, in effect, that the phylogeny of the pathogen samples and the transmission tree of the epidemic coincide. No coexistence of separate lineages within the same host is permitted which, over the short timespan of an epidemic, might not be realistic. The alternative is to treat the phylogenetic and transmission trees as separate, although related, entities, and explicitly model a phylogeny occurring within each host. The initial exploration of this was performed by Ypma et al. \cite{ypma_relating_2013}, who linked up individual within-host phylogenies according to a transmission tree structure to build a single tree describing the history of the pathogen lineages for an entire epidemic. They applied the principle to simulated measles outbreaks and data from the 2001 UK foot and mouth disease outbreak, using rather different mathematical formulations for each.  Our objective here is to build a general framework for an analysis of this sort, that is publicly available and easily modifiable for different models of host-to-host transmission, within-host pathogen population dynamics and nucleotide substitution. 

The MCMC procedure used by Ypma et al. \cite{ypma_relating_2013} treated every individual within-host phylogeny as a distinct entity and modified them individually. Two previous papers have noted that, instead, a transmission history can be reconstructed by augmenting the internal nodes of a single phylogenetic tree for the entire epidemic with information about the host in which the corresponding lineage was located. Cottam et al. \cite{cottam_integrating_2008} were the first to identify this, and it was recently revisited and refined by Didelot et al. \cite{didelot_bayesian_2014}. These studies, however, have been constrained by the lack of a method to co-estimate the complete phylogeny simultaneously with its node labels; they have instead used a fixed tree pre-generated by a standard phylogenetic method. (Another recent paper, by Vrancken et al. \cite{vrancken_genealogical_2014}, encountered the opposite difficulty, and estimated a phylogeny consistent with a fixed transmission history.) This leads to two problems. Firstly, the use of a single tree will ignore any uncertainty in estimates of the phylogeny. If a Bayesian phylogeny reconstruction method is used, this can be mitigated to some extent by using the same method on each one of a sample of trees drawn from the posterior distribution, but at the cost of greater computation time. Secondly, and more seriously, a time-resolved tree constructed using such a method will usually have been built using assumptions about the pathogen population structure that are incompatible with what we know about an epidemic. Commonly, all viral lineages are assumed to be part of a single, freely mixing population, the probability of a tree calculated based on the assumption that it was generated by a coalescent process in this population. The result is that phylogenies may display features that are not epidemiologically plausible. For example, while mutation rates for, particularly, RNA viruses are fast, it remains true that many sequences collected over the short timescale of an epidemic will be identical \cite{bataille_evolutionary_2011}. If this is the case for two isolates, they are likely to form a ``cherry'' in the reconstructed phylogeny whose time of most recent common ancestor (TMRCA) can take values very close to the sampling time of the earlier isolate, because in a panmictic population, there is no reason to rule this out. In an epidemic situation where each sample is taken from a different host, we know that this is impossible, as there must have been at least one infection event since that TMRCA, and in the time from infection to sampling, a host will have gone through an incubation period and probably also a period from manifestation of symptoms to sampling. If a single tree with these short terminal branch lengths is then used to estimate epidemiological parameters, estimates of times from infection to sampling are unlikely to be reliable.

Our contribution here is threefold. Firstly, we formally establish that the procedure for augmenting internal nodes in a phylogeny identified by Didelot et al. \cite{didelot_bayesian_2014} does indeed allow simultaneous exploration of the complete space of both phylogenies and transmission trees. Secondly, we provide a full Bayesian MCMC framework for estimation of phylogenies using a model of the pathogen population that is consistent with host-to-host transmission during an epidemic, integrating relevant epidemiological data.  Thirdly, as our method is fully integrated into the existing phylogenetics application BEAST \cite{drummond_bayesian_2012}, it provides a freely-available implementation of a method of this type for use by the research community, as well as platform for future development that has access to all the models and methods that are already implemented in that package.

\section{Method}

\subsection{Transmission trees as phylogenetic tree partitions}

We take as our dataset $D$ a set of $N$ sequences, each taken from a different infected unit (be it an infected organism or infected premise - from now on we use the word ``host'', but it need not be a single organism) in an infectious disease outbreak or epidemic, such that the total number of infections was also $N$. Let our set of hosts be $\mathbf{A}=\{a_1,\ldots,a_N\}$.  Let $\mathcal{T}$ be a genealogy describing the ancestral relationship between those $N$ isolates, with branch lengths in units of time. It consists of two components:

\begin{itemize}
\item{A rooted, binary tree $T$ with a set $\mathbf{E}_T$ of $N$ labelled tips (labelled with the elements of $\mathbf{A}$) and a set $\mathbf{I}_T$ of $N-1$ internal nodes. Let $\mathbf{N}_T=\mathbf{E}_T\cup\mathbf{I}_T$ be the complete set of nodes. Let $\Gamma_\mathbf{A}$ be the set of all such trees.} 
\item{A length function $l:\mathbf{N}_T\rightarrow(0,\infty)$ that takes each non-root node of $T$ to the difference in calendar time (in whatever units we choose) between the event represented by that node and the event represented by its parent. The event represented by an element of $\mathbf{E}_T$ is the sampling of the isolate from the host corresponding to $u$'s label; the event represented by an element of $\mathbf{I}_T$ is the existence of the most common ancestor of the isolates that correspond to $v$'s descendants. In contrast to the convention in most phylogenetic methods, we do indeed define a nonzero $l(r)$ for the root node $r$ of $T$. Its value is largely arbitrary, but it must be greater than any plausible value for the time between the existence of the event (generally an ancestor) represented by $r$ and the infection event that seeded the entire outbreak.}
\end{itemize} The length function $l$ allows us to also define a height function $h:\mathbf{N}_T\rightarrow[0,\infty)$ that takes each node to the difference in time between the event represented by that node and the time at which the last isolate was sampled.

For our purposes, we define a transmission tree on $\mathbf{A}$ to be a rooted tree with $N$ nodes labelled with the elements of $\mathbf{A}$.  The root node of such a tree is labelled with the first case in the outbreak, and the children of a node are labelled with the hosts that were directly infected by that node's label. In this framework, transmission trees do not contain timing information and consist solely of a description of which host infected which others. They are not binary and a node can have any number of children. In fact, if  $\mathcal{N}$ is such a tree, it can be thought of as a map $\mathcal{N}:\mathbf{A}\rightarrow\mathbf{A}\cup\emptyset$ taking each host $a_i$ to its infector $\mathcal{N}(a_i)$, or to $\emptyset$ if $a_i$ is the first host, and we will use this notation henceforth.

Let $\Pi_\mathbf{A}$ be the set of all transmission trees on $\mathbf{A}$. ($\Pi_\mathbf{A}$ has cardinality $N^{N-1}$ by Cayley's formula, as there are $N^{N-2}$ such trees and $N$ choices of root for each.) Take $T$ be a phylogenetic tree as above, describing the ancestry of $\mathbf{A}$, and assume no reinfection of hosts. We are interested in the set of transmission trees in $\Pi_\mathbf{A}$ that are consistent with the ancestry represented by $T$. Let $\Omega^T$ be the set of partitions of the set of nodes of $T$ such that:

\begin{itemize}
\item{If $\mathcal{P}\in\Omega^T$ and $p\in \mathcal{P}$, then the removal from $T$ of all nodes in $\mathbf{N}_T$ that are not in $p$, and all edges adjacent to at least one of them, leaves a connected graph.}
\item{All elements of $\Omega^T$ contain one and only one tip of $T$.}
\end{itemize}

For $\mathcal{P}\in\Omega^T$, define a map $\delta_\mathcal{P}:\mathbf{N}_T\to\mathbf{A}$ that takes each node of $T$ to the label of the tip that is in the same element of $\mathcal{P}$ as itself. For each $a_i\in\mathbf{A}$, let $S_{\mathcal{P},i}$ be the subtree of $T$ constructed by removing all nodes, and edges adjacent to them, that do not map to $a_i$ under $\delta_\mathcal{P}$. Because $S_{\mathcal{P},i}$ is connected, it has a single root node.  Define a second map $\epsilon_\mathcal{P}:\mathbf{A}\to\mathbf{N}_T$ taking each $a_i$ to this root node. For brevity write $s_i=\epsilon_\mathcal{P}(a_i)$. All $s_i$ have a parent $s_iP$ in $T$, except for the root $r$ of $T$ (which must be the root of one such subtree). We also define a map $\gamma:\mathbf{A}\rightarrow \mathbf{E}_T$ taking a host to the tip of $T$ which is labelled with it.

If $T$ does indeed describe the ancestral relationships between the isolates collected from the elements of $\mathbf{A}$, and we know that we have sampled every host and that there is no reinfection, it is quite intuitively clear (see figure~\ref{sameoldsameold}) that an element $\mathcal{P}$ of $\Omega^T$ corresponds to a transmission history for the epidemic. The preimage of $a_i\in\mathbf{A}$ under $\delta_\mathcal{P}$ is the set of nodes that make up $S_{\mathcal{P},i}$. Infection events occur along branches of $T$ whose start and end nodes are in different elements of $\mathcal{P}$. The assumption of no reinfection mandates the connectedness requirement (or there would be multiple introductions to the same host) and the assumption that all hosts in the outbreak were sampled mandates that each element of $\mathcal{P}$ contains a tip (because one that did not would correspond to an unsampled host).

To formalise the correspondence, we construct a map $z:\Omega^T\rightarrow\Pi_\mathbf{A}$ such that if $\mathcal{P}\in\Omega^T$ and $a_i\in\mathbf{A}$, 
\begin{eqnarray*}
z(\mathcal{P})(a_i) = \begin{cases} 
\delta_\mathcal{P}(s_iP) & s_i\neq r\\
\emptyset & s_i= r
\end{cases}
\end{eqnarray*}

\begin{proposition}
For $\mathcal{P}\in\Omega_T$, the directed graph given by drawing an edge from $z(\mathcal{P})(a_i)$ to $a_i$ for all $a_i\in\mathbf{A}$ is a tree, and if $r$ is the root of $T$, the directionality coincides with that given by taking $\delta_\mathcal{P}(r)$ to be its root. 
\end{proposition}

\begin{proof}
For the first part, we must show that the graph is simple, connected, and has no cycles. For simplicity, the construction will never give a node with indegree greater than 1, so if two edges join the same two nodes then their directionality is different. Suppose $a_i,a_j\in\mathbf{A}$ are such that $a_i=\delta_\mathcal{P}(s_jP)$ and $a_j=\delta_\mathcal{P}(s_iP)$. Now $s_i$ and $s_jP$ (which may not be distinct) are nodes of $S_{\mathcal{P},i}$, and $s_j$ as a descendant of $s_jP$ is also a descendant of $s_i$ in $T$. Similarly, $s_i$ is a descendant of $s_j$. This contradicts the fact that $T$, as a tree, has no cycles, or, if $s_j=s_iP$ and $s_i=s_jP$, that it is simple.

For connectedness, again suppose $a_i\in\mathbf{A}$ and let $a_j=\delta_\mathcal{T}(r)$; the root $a_j$ of $S_{\mathcal{P},j}$ is the root of $T$. It may be that $a_i=a_j$. If not, the path in $T$ from $a_i$ to $a_j$ passes through $n\geq2$ elements of $\mathcal{P}$ whose elements map under $\delta_\mathcal{P}$ to the hosts $a_{o(1)},\ldots,a_{o(n)}\in\mathbf{A}$, where $o$ is some permutation of $\{1,\ldots,N\}$ with $o(1)=i$ and $o(n)=j$. In particular it must pass through the root nodes of all these subtrees, $s_{o(1)},\ldots,s_{o(n)}$, implying that $z(\mathcal{P})(a_{o(k)})=a_{o(k+1)}$ for all $1\leq k \leq n-1$. It follows that $(z(\mathcal{P}))^{n-1}(a_i)=a_j$; thus all hosts in $\mathbf{A}$ are connected to $a_j$ and each other.

Suppose $z(\mathcal{P})$ has a cycle. It must be a directed cycle or else $z(\mathcal{P})$ has a node with indegree greater than 1. With $o$ denoting a permutation of $\{1,\ldots,N\}$ as before, suppose the cycle has $n\geq 3$ (if $n=2$ then the graph is not simple) elements $a_{o(1)},\ldots,a_{o(n)}$ such that $z(\mathcal{P})(a_{o(k)})=a_{o(k+1)}$ for all $1\leq k \leq n-1$ and $z(\mathcal{P})(a_{o(n)})=a_{o(1)}$. If $i\geq2$, the $S_{\mathcal{P},o(i)}$ is a subtree of $T$ containing a root node $s_{o(i)}$ and the parent $s_{o(i-1)}P$ of the root node of the subtree $S_{\mathcal{P},o(i-1)}$; similarly $S_{\mathcal{P},o(1)}$ contains $s_{o(n)}P$. Since $S_{\mathcal{P},o(i)}$ for each $i$ contains a sequence of nodes, following the directedness of $T$ induced by its root, running from $s_{o(i)}$ to $s_{o(i-1)}P$ to and there is a directed link from each $s_{o(i)}P$ to $s_{o(i)}$ in $T$, the concatenation of all of these forms a cycle in $T$, contradicting the fact it is a tree.

For the second part, there is no node $z(\mathcal{P})(\delta_\mathcal{P}(r))$ by construction, and we have already shown that our construction produces a directed path from each $a\in\mathbf{A}$ to $\delta_\mathcal{P}(r)$. As we have shown $z(\mathcal{P})$ is a tree, this is the only such path, hence the directedness of all edges is towards $\delta_\mathcal{P}(r)$.
\end{proof}

\begin{proposition}\label{inj}
$z$ is injective.
\end{proposition}

\begin{proof}
We suppose the we have two partitions $\mathcal{P},\mathcal{P}'$ that have the same image under $z$, i.e. for all $a_i\in\mathbf{A}$, $z(\mathcal{P})(a_i)=z(\mathcal{P}')(a_i)$. If $\mathcal{P}\neq\mathcal{P}'$ then there exists some node $u$ of $T$ that has $a_i=\delta_\mathcal{P}(u)\neq a_j=\delta_{\mathcal{P}'}(u)$. We can assume that either $u$ is the root of $T$ or $\delta_\mathcal{P}(uP)=\delta_\mathcal{P'}(uP)$ for the parent $uP$ of $u$ (or else we move down $T$ to find a new $u$ for which this is true). 

If $u$ is the root of $T$, then it is the root of the subtrees $S_{\mathcal{P},i}$ and $S_{\mathcal{P}',j}$. This implies $z(\mathcal{P})(a_i)=\emptyset$ but $z(\mathcal{P}')(a_i)\neq\emptyset$ because $z(\mathcal{P}')(a_j)=\emptyset$; only one element of $\mathbf{A}$ can be sent to $\emptyset$ by $z(\mathcal{P}')$ since the root of $T$ is unique. So $uP$ exists.

Let $a_k=\delta_\mathcal{P}(uP)=\delta_\mathcal{P'}(uP)$.  First suppose $k\neq i$ and $k\neq j$. Then $z(\mathcal{P})(a_i)=a_k$. We show that $z(\mathcal{P'})(a_i)=a_k$ is not possible. Let $v=\gamma(h)$. Now $v$ is a descendant of $u$ because $u$ is the root node of the subtree $S_{\mathcal{P},i}$, and $S_{\mathcal{P},i}$ includes $v$. $\mathcal{P'}$ gives rise to another subtree of $T$, $S_{\mathcal{P}',i}$, all of whose nodes map to $a_i$ under $\delta_{\mathcal{P}'}$. This $S_{\mathcal{P}',i}$ has a root node $s'_i$ which is \emph{not} $u$ because $\delta_{\mathcal{P}'}(u)=a_j$. It must, in fact, also be a descendant of $u$; if it were not, $S_{\mathcal{P}',i}$ would be disconnected by $u$. The parent $s'_iP$ cannot have $\delta_{\mathcal{P}'}(s'_i)=a_k$ because either a) $s'_iP=u$ and $\delta_{\mathcal{P}'}(u)=a_j$ by construction or b) $s'_iP\neq u$ and if $\delta_{\mathcal{P}'}(s'_iP)=a_k$ were true, the subtree of nodes that map to $a_k$ under $\delta_{\mathcal{P}'}$ would be disconnected by $u$. Hence $z(\mathcal{P'})(a_i)\neq a_k$.

So without loss of generality suppose $k\neq i$ but $k= j$. Again $z(\mathcal{P})(a_i)=a_k$. Let $v$ be the unique tip of $T$ that has $\delta_\mathcal{P}(v)=\delta_\mathcal{P'}(v)=a_k$. Now, $v$ is not a descendant of $u$. If it were, then $S_{\mathcal{P},k}$, the subtree of $T$ whose nodes are mapped to $a_k$ by $\delta_{\mathcal{P}}$, would be disconnected by $u$, which maps to $a_i$. This implies that there is a descendant $w$ of $u$ in $T$, possibly $u$ itself, which maps to $a_k$ under $\delta_{\mathcal{P}'}$ but neither of whose children $wC_1$ and $wC_2$ do. (If this were not true, a second tip would map to $a_k$ under $\delta_{\mathcal{P}'}$). Whether it is $u$ or not, $w$ cannot map to $a_k$ under $\delta_\mathcal{P}$; if it is $u$ then it does not by construction, and if is not, it would have an ancestor, $u$, which did not, and an earlier ancestor, $uP$, which did, breaking connectedness. This implies that $z(\mathcal{P}')(wC_1)=z(\mathcal{P}')(wC_2)=a_k$ but $z(\mathcal{P})(wC_1)=z(\mathcal{P})(wC_2)\neq a_k$.

\end{proof}

For the next proposition, we need the following:

\begin{lemma}\label{ancestors}
If $a_i,a_j\in\mathbf{A}$ and $\mathcal{N}\in\Pi_\mathbf{A}$ is a transmission tree in which $a_i$ is an ancestor of $a_j$, then if $\mathcal{P}\in\Omega^T$ with $z(\mathcal{P})=\mathcal{N}$ and $u$ is a node of $T$ with $\delta_\mathcal{P}(i)=a_j$, $u$ has an ancestor $v$ in $T$ with $\delta_\mathcal{P}(j)=a_i$.
\end{lemma}

\begin{proof}
Strong induction on the number $n$ of intervening hosts between $a_i$ and $a_j$ in $\mathcal{N}$. If $n=0$, this is true by definition of $u$, as the node $r_{h_2}$ is an ancestor of $i$ and its parent maps to $h_1$.  If the lemma is true for all $n\leq m$ and the set of intervening hosts has size $m+1$, let $a_k$ be an arbitrary member of that set. The number of intervening hosts between $a_k$ and $a_j$ in $\mathcal{N}$ is less than $m+1$, so $i$ has an ancestor $v$ in $T$ with $\mathcal{P}(k)=a_k$. The number of intervening hosts between $a_i$ and $a_k$ in $\mathcal{N}$ is also less than $m+1$, so $v$ has an ancestor $w$ in $T$ with $\mathcal{P}(k)=a_i$. It follows that $w$ is the ancestor of $u$ that we need.
\end{proof}

\begin{proposition}\label{notsurj}
$z$ is not surjective for $N>2$. 
\end{proposition}

\begin{proof}
For $N=2$, $|\Pi_\mathbf{A}|=2$ and $|\Omega^T|=2$ since the latter is simply the number of assignments for the single internal node of $T$ to a subgraph containing one tip or the other. The map's injectiveness ensures its surjectiveness. If $N>2$, then let $a_i,a_j,a_k\in\mathbf{A}$ be any three hosts. In $T$, $\gamma(a_i)$, $\gamma(a_j)$ and $\gamma(a_k)$ have a most recent common ancestral node $u$ and two of them, without loss of generality $\gamma(a_j)$ and $\gamma(a_k)$, have a most recent common ancestral node $v$ which is a descendant of $u$. We show that there is no element of $\Omega^T$ which will map to any member of $\Pi_\mathbf{A}$ in which any of the following are true:

\begin{itemize}
\item{$a_j$ is an ancestor of $a_i$, which is an ancestor of $a_k$.}
\item{$a_j$ is an ancestor of $a_k$, which is an ancestor of $a_i$.}
\item{$a_k$ is an ancestor of $a_i$, which is an ancestor of $a_j$.}
\item{$a_k$ is an ancestor of $a_j$, which is an ancestor of $a_i$.}
\end{itemize}

Let $\mathcal{P}$ be a partition such that $z(\mathcal{P})$ is a transmission tree in which $a_j$ is an ancestor of both $a_i$ and $a_k$. Now $\delta_{\mathcal{P}}(u)=a_j$. To see this, note that since $u$ is an ancestor of $\gamma(a_j)$, if it does not map to $a_j$ under $\delta_{\mathcal{P}}$ then neither do any of its ancestors, by connectedness. Nor do any descendants of the child of $u$ which is not an ancestor of $\gamma(a_j)$ and $\gamma(a_k)$, a set which includes $\gamma(a_i)$. All ancestors of $\gamma(a_i)$ apart from $u$ belong to one of those categories. But this contradicts lemma~\ref{ancestors} because $\gamma(a_i)$ has no ancestor which maps to $a_j$ under $\delta_\mathcal{P}$ despite the fact that $a_j$ is an ancestor of $a_i$.

Now $\gamma(a_i)$ has no ancestor in $T$ that maps to $a_k$ under $\delta_{\mathcal{P}}$, because the node $u$ breaks connectedness between $\gamma(a_k)$ and any position that such a node could be. The contrapositive of lemma~\ref{ancestors} then says that $a_k$ is not an ancestor of $a_i$. Similarly $a_i$ is not an ancestor of $a_k$. Likewise, if $z(\mathcal{P})$ is such that $a_k$ is an ancestor of both $a_i$ and $a_k$, $a_i$ is not an ancestor of $a_j$ nor vice versa.

\end{proof}

Let the image of $\Omega^T$ under $z$ be $\Lambda^T_\mathbf{A}\subseteq\Pi_\mathbf{A}$. The actual cardinality of $\Lambda^T_\mathbf{A}$ varies with the topology of $T$, which can be clearly seen in the case $N=4$ (figure~\ref{4partitions}).  

Proposition~\ref{inj} states that no two partitions of the internal nodes of $T$ correspond to the same transmission history; the set of partitions and the set of compatible transmission trees are equivalent. Proposition~\ref{notsurj} shows, however, that not every possible transmission tree on $\mathbf{A}$ actually corresponds to a partition of the nodes of a fixed $T$. If we are interested in exploring the complete space of transmission trees using this construction, we need to vary the phylogeny as well.

Let the set $\mathbf{\Omega}=\{\Omega^T:T\in\Gamma_\mathbf{A}\}$ consist of all partitions of all phylogenies with tips labelled with $\mathbf{A}$. The map $z$ can be extended to a map $Z:\mathbf{\Omega}\to\Pi_\mathbf{A}$ in the obvious way.

\begin{proposition}\label{zsurj}
$Z$ is surjective. In other words, any transmission tree on $\mathbf{A}$ arises as a partition of \emph{some} phylogenetic tree $T\in\Gamma_\mathbf{A}$.
\end{proposition}
 
\begin{proof}

Let $\mathcal{N}\in\Pi_\mathbf{A}$. Use the following procedure to construct an element of $\mathbf{\Omega}$. If each $a_i\in\mathbf{A}$ has $n_i$ children in $\mathcal{N}$, take $n_i+1$ nodes $v_{i,1},\ldots,v_{i,n_i},v_{i,n_i+1}$. Pick an arbitrary ordering of the children of each $a_i$ and make a graph $T$ by drawing two edges from each $v_{i,k}$ to $v_{i,k+1}$ and from $v_{i,k}$ to $v_{j,1}$ where $j$ is such that $a_j$ is the $k$th child of $a_i$ in the ordering. (Notice that $v_{i,n+1}$ gets no children either way.) If $r\in\{1,\ldots,N\}$ is such that $a_r$ is the root of $\mathcal{N}$, let the root of $T$ be $v_{r,1}$.

It is clear that $T$ is a rooted binary tree, its tips are the $v_{i,n_i+1}$ and if each of these is labelled with the corresponding $a_i$ then they are in one-to-one correspondence with $\mathbf{A}$. The set of nodes $v_{i,1},\ldots,v_{i,n_i},v_{i,n_i+1}$ for each $h_i$ are by construction connected in $T$ and contain the single tip $v_{i,n_i+1}$; hence this partitioning of the nodes of $T$ is an element $\mathcal{P}$ of $\Omega^T$.  It is easily checked that $z(\mathcal{P})=\mathcal{N}$.
\end{proof} 
 
As an aside, $Z$ is not injective, as is clear from the arbitrary choice of ordering for the children of each $a_i$. (In fact, some elements of $\mathbf{\Omega}$ cannot be produced by this construction at all, for example, the bottom right example in figure~\ref{sameoldsameold}.) The upshot of proposition~\ref{zsurj} is that a MCMC procedure that fully explores the space of these partitioned phylogenies is also fully exploring the space of transmission trees amongst the elements of $\mathbf{A}$. We outline such a procedure in the next section.

So far, we have only dealt with the phylogenetic tree topology $T$. If this construction is to be useful for epidemic reconstruction, we must now consider branch lengths. Let $\mathcal{P}$ be a partition of $T$, and suppose $T$ is the topology of a genealogy $\mathcal{T}$ with length function $l$ and height function $h$. Suppose $a_i\in\mathbf{A}$ and that $z(\mathcal{P})(a_i)\neq\emptyset$. Let $u=\epsilon_\mathcal{P}(a_i)$, and let $uP$ be the parent of $u$. An infection event occurs on the branch between $uP$ and $u$, which means, assuming that internal nodes of $T$ and transmissions do not occur at exactly the same time, that it occurs at a height in the interval $(h(u),h(uP))$. In what follows it will be convenient to use a forwards timescale, so let $C:\mathbb{R}\to\mathbb{R}$ be a function converting between tree height and such a timescale (in the same units, so branch lengths are maintained). Let $t^{\textrm{inf}}_i$ be this time of infection in \emph{forwards} time. Let $q_i\in(0,1)$ be such that $t^{\textrm{inf}}_i = C(h(uP)) + q_i(C(h(u))-C(h(uP))) = C(h(uP)) + q_il(u)$. If $z(\mathcal{P})(a_i)=\emptyset$, i.e. $a_i$ is the first host in the epidemic, then $t^{\textrm{inf}}_i$ is between $C(h(r)+l(r))$ ($r$ being the root node of $T$) and $C(h(r))$ (remembering that we gave $r$ a finite branch length) we can similarly define $q_i$ such that $t^{\textrm{inf}}_i = C(h(r)+l(r)) + q_il(r)$.

The combination of a genealogy $\mathcal{T}$, partition $\mathcal{P}$ and a set of $q_i$s for all elements $a_i\in\mathbf{A}$ then entirely determines the transmission history of the epidemic, describing which host infected which others and when. No assumptions are made at this, conceptual, stage about when hosts cease to be infectious; a host can continue to infect others at any time following the time at which is sample was acquired. If, as will often be the case, this is an unreasonable assumption, the likelihood of such partitions can be evaluated to zero in the calculation of the posterior probability.

\section{MCMC procedure}

The most common methods for estimation of time-resolved phylogenies involve the use of Bayesian MCMC to sample from the probability distribution of phylogenetic trees given the available sequence data. The previous section demonstrates that, if the sequence data is such that one sample is taken from each host, such procedures can be extended to simultaneously sample from the probability distribution of reconstructed epidemics each sampled tree is augmented a partition of its nodes as well as the values of each $q_i$. We have implemented this procedure in the package BEAST \cite{drummond_bayesian_2012}. Because of the special requirements of this type of augmentation, the standard moves on the phylogenetic tree topology cannot be used. Nor are the structured tree operators developed by Vaughan et al. \cite{vaughan_efficient_2014} suitable, as those are designed for the exploration of the space of trees where every point on every branch can be freely assigned a ``type'' from a finite set. This condition is much less restrictive than than connectedness requirements that we have outlined above and the result of such a move on a tree of our type would not necessarily meet our requirements for partitions. Instead, specialised moves have been devised to alter the partitioned phylogeny in such a way that the transmission tree structure is maintained. In addition, we give an operator to alter the transmission tree while keeping the phylogenetic tree fixed, by changing node labels. 

Note that these moves do not simultaneously change the value of any of the $q_i$s, as moves on these are proposed and evaluated separately. Nevertheless, changes to either tree may involve resampling the times of infection of some hosts. If $a_i\in\mathbf{A}$, changing partition from $\mathcal{P}$ to $\mathcal{P}'$ may mean that $\epsilon_\mathcal{P}(a_i)$ and $\epsilon_{\mathcal{P}'}(a_i)$ are different nodes with different heights, and so while $q_i$ will not change, $t^{\mathrm{inf}}_i$ will. Even a move has no effect on the partition or phylogenetic tree topology, such as a change to branch lengths, may also alter the height of $\epsilon_\mathcal{P}(a_i)$ and/or its parent, which will also modify $t^{\mathrm{inf}}_i$ while $q_i$ remains fixed.

\begin{definition}
For a partition $\mathcal{P}$ of a phylogeny $\mathcal{T}$, if $u$ is a phylogenetic tree node with $\delta_\mathcal{P}(u)=a_i\in\mathbf{A}$ we say $u$ is \emph{ancestral under $\mathcal{P}$} if it is an ancestor of the only member of the subtree $S_{\mathcal{P},i}$ which is a tip of $\mathcal{T}$.
\end{definition}

\begin{definition}
For a partition $\mathcal{P}$ of a phylogeny $\mathcal{T}$, the \emph{infection branch} for $a_i\in\mathbf{A}$ is the branch of $\mathcal{T}$ ending in $\epsilon_\mathcal{P}(a_i)$.
\end{definition}

\subsection{Infection branch operator}
We randomly select a host $a_i$ that is not the first host in the outbreak (i.e. $\epsilon_\mathcal{P}(a_i)$ is not the root of $\mathcal{T}$). Consider $\epsilon_\mathcal{P}(a_i)$. The operator performs both ``downward'' and ``upward'' moves, but if $\epsilon_\mathcal{P}(a_i)$ is a tip then the move must be downwards. If it is internal, then we select upwards or downwards each with probability 0.5. Let $u=\epsilon_\mathcal{P}(a_i)$ and $uP$ be the parent of $u$ (which must exist as we avoided the root). It must be that $u$ and $uP$ are in different elements of $\mathcal{P}$, and this implies that $u$ is ancestral under $\mathcal{P}$ because the path from any node $v$ that is not a descendant of $u$ to $u$ must pass through $uP$ and if $\delta_\mathcal{P}(v)=a_i$ this would violate the connectedness requirement. Suppose $\delta_{\mathcal{P}}(uP)=a_j$.

\paragraph{Upward move} We create a new partition $\mathcal{P}'$ that has $\delta_{\mathcal{P}'}(u)=a_j$, moving the infection branch of $a_i$ up the tree. Consider the two children $uC_1$ and $uC_2$ of $u$ (as this is the upward move, $u$ is not a tip). At least one of these is mapped to the same element of $\mathbf{A}$ as $u$ by $\delta_\mathcal{P}$ because $u$ must be in the same element of $\mathcal{P}$ as the tip $\gamma\circ\delta_\mathcal{P}(u)$ and the path from $u$ to this tip in the subtree will intersect one of its children. If this is true of only one child then without loss of generality say it is $uC_1$. In this case we can simply make $\mathcal{P}'$ by setting $\delta_{\mathcal{P}'}(i)=a_j$ and leaving the rest of the partition unchanged; this is clearly still a valid partition because all subtrees remain connected. So suppose also $\delta_{\mathcal{P}}(uC_2)=\delta_{\mathcal{P}}(u)$. At most one of $uC_1$ and $uC_2$ is ancestral under $\mathcal{P}$ (as siblings, they cannot both be ancestors of the same tip) so, again without loss of generality, say it is $uC_1$. If we again set $\delta_{\mathcal{P}'}(u)=a_j$, the removal of $u$ from the subtree $S_{\mathcal{P},i}$ splits the nodes of the latter into two sets, $V_1$ containing $uC_1$ and $\gamma\circ\delta_\mathcal{P}(u)$, and $V_2$ containing $uC_2$. The nodes of both sets and the edges between them form connected subtrees of $T$, but their union is not connected.  We complete the construction of $\mathcal{P}'$ by setting $\delta_{\mathcal{P}'}(v)=a_j$ for all $v\in V_2$. $S_{\mathcal{P}',i}$ and $S_{\mathcal{P}',j}$ are then connected. 

The effect on the transmission tree is that all $a_k\in\mathbf{A}$ that have $z(\mathcal{P})(a_k)=a_i$ and $\gamma(a_k)$ a descendant of $uC_2$ have $z(\mathcal{P}')(a_k)=a_j$ instead.

\paragraph{Downward move} We create a new partition $\mathcal{P}'$ that has $\delta_{\mathcal{P}'}(uP)=a_i$, moving the infection branch of $a_i$ down the tree. We need to consider the grandparent $uG$ of $u$ if it exists, and the child $uS$ of $uP$ that is not $u$. At least one of $uG$ and $uS$ must be in the same element of $\mathcal{P}$ as $uP$ (or else $uP$ is not in a partition element containing a tip). If $uG$ does not exist then this must be $uS$.

If $\delta_\mathcal{P}(uS)=a_j$ and either $\delta_\mathcal{P}(uG)\neq a_j$ or $uG$ does not exist, then setting $\delta_{\mathcal{P}'}(uP)=a_i$ is all that is required to make $\mathcal{P}'$ a valid partition. The two or three nodes joined to $uP$ by edges were all in different elements of $\mathcal{P}$ and remain so; $uP$ was in the element of $\mathcal{P}$ containing one of its children and is moved to the one containing the other child in $\mathcal{P}'$. Similarly, if $\delta_\mathcal{P}(uG)=a_j$ and $\delta_\mathcal{P}(uS)\neq\delta_\mathcal{P}(uP)$, then $\mathcal{P}'$ then all we need do is set $\delta_{\mathcal{P}'}(uP)=a_i$; the situation is the same except that the $uP$ has moved from the element of $\mathcal{P}$ that contains of one of its children to the one containing its parent.

If $uG$ exists and $\delta_\mathcal{P}(uS)=\delta_\mathcal{P}(uG)=a_j$, then the removal of $uP$ from the subtree $S_{\mathcal{T},j}$ splits into two subtrees whose union is again not a connected subtree of $T$. Let the node sets of these two subtrees be $V_1$ and $V_2$, with $V_1$ containing $uG$ and $V_2$ containing $uS$. If $uP$ is ancestral under $\mathcal{P}$ then $V_2$ also contains the tip $\gamma(a_j)$, and if it is not then $V_1$ does. We complete $\mathcal{P}'$ by setting $\delta_{\mathcal{P}'}(uS)=a_i$ for all $v$ in the set that does not contain $\gamma(a_j)$. $S_{\mathcal{P}',i}$ and $S_{\mathcal{P}',j}$ are now connected. Note that $V_1$ may contain the root node and if it does not contain $\gamma(a_j)$ then the root's image under $\delta_\mathcal{P}$ is different from that under $\delta_{\mathcal{P}'}$, which is how this move may change the first host in the outbreak even though the root host is never chosen by the move. This can be seen in example 7) of figure~\ref{ttoperator}.

If $uP$ is not ancestral under $\mathcal{P}$, then the effect on the transmission tree is that all $a_k\in\mathbf{A}$ that have $z(\mathcal{P})(a_k)=a_j$ and $\gamma(a_k)$ a descendant of $uS$ have $z(\mathcal{P}')(a_k)=a_i$ instead. If $uP$ is ancestral under $\mathcal{P}$ then, in $z(\mathcal{P}')$, $a_i$ is the infector of $a_j$ instead of vice versa, and all $a_k\in\mathbf{A}$ that have $z(\mathcal{P})(a_k)=a_j$ and $\gamma(a_k)$ \emph{not} a descendant of $uS$ have $z(\mathcal{P}')(a_k)=a_i$ instead. 

\paragraph{Hastings ratio} We observe that:

\begin{itemize}
\item{The upward move on $u$ is reversed by the downward move on the child $uC_1$ of $u$ that is ancestral under $\mathcal{P}$. Thus the Hastings ratio is 1 if $uC_1$ is not a tip and 2 if it is.}
\item{If $uP$ is not ancestral under $\mathcal{P}$, then the downward move on $u$ is reversed by the upward move on $uP$. The Hastings ratio is 1 if $u$ is not a tip and $1/2$ if it is.}
\item{If $uP$ is ancestral under $\mathcal{P}$, then the downward move on $u$ is reversed by the downward move on its sibling $uS$. The Hastings ratio is 1 if $u$ and $uS$ are both tips or both not tips, 1/2 if $u$ is but $uS$ is not, and 2 if $uS$ is but $u$ is not. }
\end{itemize}

The various variations of this move are depicted in figure~\ref{ttoperator}. If the initial partition is that depicted as 1), the downward moves depicted as 2), 4), 7), 9), 10) and 12) involve a parent that is ancestral under $\mathcal{P}$ and 5) and 6) involve one that is not.

\subsection{Phylogenetic tree operators}

We have adapted the three standard tree moves used in BEAST (exchange, subtree slide, and Wilson-Balding \cite{wilson_genealogical_1998, drummond_estimating_2002, hohna_clock-constrained_2008}) such that they respect the transmission tree structure induced by partitioning the internal nodes. We give two versions of each:

\begin{itemize}
\item{A ``type A'' operator which does not alter the transmission tree at all; all parental relationships remain the same.}
\item{A ``type B'' operator which performs phylogenetic tree modifications which simultaneously rearrange the transmission tree by assigning new parents to one or two hosts.} 
\end{itemize}

\subsubsection{Type A operators}

\paragraph{Type A exchange}

Select a random node $u$ that is not the root $r$ of the phylogenetic tree $\mathcal{T}$, and then randomly selects a second node $v$, also not $r$ and not the sibling $uS$ of $u$, such that the parents $uP$ and $vP$ of $u$ and $v$ are in the same element of $\mathcal{P}$, $h(uP)>h(v)$, and $h(vP)>h(u)$. If there is no such $v$ then the operator fails. Otherwise, $u$ and $v$ exchange parents to obtain a new phylogenetic tree $\mathcal{T}'$ with the same partition of nodes $\mathcal{P}$. $\mathcal{P}$ is still valid in terms of connectedness, because if $\delta_\mathcal{P}(u)\neq\delta_\mathcal{P}(uP)$ then all nodes in the element of $\mathcal{P}$ containing $u$ are descendants of $u$ and the move has not affected them, whereas if $\delta_\mathcal{P}(u)=\delta_\mathcal{P}(uP)$ then changing $u$'s parent to $vP$ means that after the move it is still adjacent to a node with the same image under $\delta_\mathcal{P}$ as itself; the same goes for $v$. The transmission tree structure is unchanged: if $\delta_\mathcal{P}(u)\neq\delta_\mathcal{P}(uP)$ then $\delta_\mathcal{P}(u)$ is infected by $\delta_\mathcal{P}(uP)$ before the move and is by $\delta_\mathcal{P}(vP)=\delta_\mathcal{P}(uP)$ afterwards, whereas if $\delta_\mathcal{P}(u)=\delta_\mathcal{P}(uP)$ then $\delta_\mathcal{P}(u)$'s infection branch was not affected at all. Again, the same goes for $v$.

For the Hastings ratio, note that the partitioned tree obtained by selecting $u$ and then $v$ is exactly the same as that obtained by selecting $v$ and then $u$. If a node $w$ is selected first, let $c_\mathcal{P}(w)$ be the number of eligible nodes to be selected as the second (this is explicitly calculated every time the operator acts). The denominator of the Hastings ratio is then $\frac{1}{2N-2}(\frac{1}{c_\mathcal{P}(u)}+\frac{1}{c_\mathcal{P}(v)})$. The move is reversed by selecting the same two nodes again (in either order) hence we calculate $c_{\mathcal{P}'}(u)$ and $c_{\mathcal{P}'}(v)$ and the ratio's numerator is $\frac{1}{2N-2}(\frac{1}{c_{\mathcal{P}'}(u)}+\frac{1}{c_{\mathcal{P}'}(v)})$. Cancellation gives  $\frac{\frac{1}{c_{\mathcal{P}'}(u)}+\frac{1}{c_{\mathcal{P}'}(v)}}{\frac{1}{c_\mathcal{P}(u)}+\frac{1}{c_\mathcal{P}(v)}}$.

\paragraph{Type A subtree slide} Select a random node $u$ under the conditions that $u\neq r$ and either $u$'s grandparent $uG$ or sibling $uS$ (or both) is in the same element of $\mathcal{P}$ as its parent $uP$. Draw a distance $\Delta\in\mathbb{R}$ from some probability distribution that is symmetric about 0. We aim to change the height of $uP$ to $h(uP)+\Delta$. If $\Delta>0$, examine $uP$'s ancestors to find a node $v$ such that either $v=r$ or $h(v)<h(uP)+\Delta$ but $h(vP)>h(uP)+\Delta$; if no such ancestor exists then let $v=uS$ and this is true. If $\delta_\mathcal{P}(v)\neq\delta_\mathcal{P}(uP)$ then the move fails.  If $v=uS$ then simply change the height of $uP$ to $h(uP)+\Delta$ and the topology is unchanged. Otherwise, modify the tree such that $uP$ has height $h(uP)+\Delta$, parent $vP$ (or no parent if $v=r$ in which case $uP$ is now the root node) and child $v$, and $uS$ has parent $uG$. Again, do not change $\mathcal{P}$.  Connectedness rules are still obeyed because, in the new tree $\mathcal{T}'$, $uP$ is adjacent to $v$, which is in the same element of $\mathcal{P}$ as itself. The transmission tree structure is unchanged as:

\begin{itemize}
\item{The move does not change the partition, so any infection branches have not changed if the particular phylogenetic tree branch was not modified by the move. This applies to the branch between $u$ and $uP$ as well as all branches adjacent to nodes other than $u$, $uP$, $uG$, $uS$, $v$, and $vP$.}
\item{If $uS$ and $uP$ are in different elements of $\mathcal{P}$ then $uP$ and $uG$ are in the same one, so the infector of $\delta_\mathcal{P}(uS)$ remains the same.}
\item{If $uG$ and $uP$ are in different elements of $\mathcal{P}$ then the move fails if $h(uP)+\Delta>h(uG)$ so the phylogenetic tree topology is unchanged.}
\item{If $v$ and $vP$ are in different elements of $\mathcal{P}$ then $uP$, instead of $v$, is now the end of $\delta_\mathcal{P}(uP)$'s infection branch, but $\delta_\mathcal{P}(uP)=\delta_\mathcal{P}(v)$ and its infector is still $\delta_\mathcal{P}(vP)$.}
\end{itemize}

If $\Delta<0$, then if $h(uP)+\Delta<h(u)$ the move fails. Otherwise, we select a node $v$ at random from the set $W$ which consists of nodes $w$ that:

\begin{enumerate}
\item{Are descendants of $uP$ but not descendants of $u$.}
\item{Have $h(k)<h(uP)+\Delta$ but $h(kP)>h(uP)+\Delta$.}
\item{Have $\delta_{\mathcal{P}}(wP)=\delta_{\mathcal{P}}(uP)$.}
\end{enumerate} 

If $W$ is empty the move fails. In the case that $W$ consists only of $uS$ then simply set $h(uP)=h(uP)+\Delta$ and the topology is unchanged. Otherwise, modify the tree such that $uP$ has height $h(uP)+\Delta$, parent $vP$ and child $v$, and $uS$ has parent $uG$. connectedness rules are still obeyed because there is an edge from $uP$ to a node ($vP$) in the same element of the partition. The transmission tree structure is unchanged as:

\begin{itemize}
\item{Again, the move does not change the partition, so any infection branches have not changed if the particular phylogenetic tree branch was not modified by the move.}
\item{If $uS$ and $uP$ are in different elements of $\mathcal{P}$ then the move fails if $h(uP)+\Delta<h(uS)$ so the topology is unchanged.}
\item{If $uG$ and $uP$ are in different elements of $\mathcal{P}$ then $uP$ and $uS$ were in the same one, so the infector of $\delta_\mathcal{P}(uP)$ remains the same; $uS$ is now the end of its infection branch.}
\item{If $v$ and $vP$ are in different elements of $\mathcal{P}$ then the infector of $\delta_\mathcal{P}(v)$ is still $\delta_\mathcal{P}(vP)=\delta_\mathcal{P}(uP)$.}
\end{itemize}

Suppose there are $d_{\mathcal{T}}$ nodes eligible for this move before it occurs and $d_{\mathcal{T}'}$ afterwards. If the topology did not change then the Hastings ratio is $\frac{d_{\mathcal{T}'}}{d_{\mathcal{T}}}$. Otherwise, it is $\frac{|W|d_{\mathcal{T}'}}{d_{\mathcal{T}}}$ if $\Delta<0$ and $\frac{d_{\mathcal{T}'}}{|W'|d_{\mathcal{T}}}$ if $\Delta>0$, where the $W'$ is the set of nodes $w$ that:

\begin{enumerate}
\item{Are descendants of $vP$ (in the original tree) but not descendants of $u$.}
\item{Have $h(w)<h(uP)$ but $h(wP)>h(uP)$.}
\item{Have $\delta_{\mathcal{P}}(wP)=\delta_{\mathcal{P}}(v)$.}
\end{enumerate} 

\paragraph{Type A Wilson-Balding move} Pick a node $u$ under the same conditions as for the type A subtree slide. Pick a second node $v$ at random from amongst all nodes that are in the same element of $\mathcal{P}$ as $uP$, or whose parents are, and such that $h(vP)>h(u)$. The move fails if $uP=vP$, or $v=uP$. The node $uP$ is pruned and reattached as a child of $vP$ and the parent of $v$ as with the standard Wilson-Balding move \cite{wilson_genealogical_1998, drummond_estimating_2002}. As before, do not change $\mathcal{P}$. Connectedness rules are obeyed because there is an edge from $uP$ to a node (either $v$ or $vP$) in the same element of $\mathcal{P}$ as itself. The transmission tree structure is unchanged because if there was an infection event between $uG$ and $uC$ (and there was at most one by construction) then there still is and it involves the same hosts, and likewise if there was one between $vP$ and $v$ then there still is and it involves the same hosts. If there was no infection event in either case then the removal or insertion of $uP$ does not add one.

Notice that if $u$ is subsequently selected for this move again, then the set of candidates for the second node is the same except that it excludes the original $v$ and includes the original $uG$; in particular it has the same cardinality, as it did for the standard Wilson-Balding move. So only the choice of first node affects the Hastings ratio. It follows that this is the ratio from the standard Wilson-Balding move multiplied by $\frac{e_{\mathcal{T}}}{e_{\mathcal{T}'}}$, where $e_{\mathcal{T}}$ is the number of nodes eligible for this move before it occurs and $e_{\mathcal{T}'}$ is the number afterwards.

\subsubsection{Type B operators}

\paragraph{Type B exchange}

Select a random node $u$, not $r$, whose parent $uP$ is in a different element of $\mathcal{P}$ to itself. Pick a second node $v$, also not $r$ and not $uS$, whose parent $uP$ is also in a different element of $\mathcal{P}$ to itself (but this time the elements containing $uP$ and $vP$ do not have to be the same), such that $h(uP)>h(v)$, and $h(vP)>h(u)$. If there is no such $v$ then the operator fails. Otherwise, $u$ and $v$ exchange parents as with the type A operator. That it preserves connectedness of subtrees is clear. The effect on the transmission tree is that $\delta_\mathcal{P}(u)$ and $\delta_\mathcal{P}(v)$ exchange parents (if their parents are different). 

The Hastings ratio is calculated in effectively the same way as for the type A version, noting that the number of choices for $u$ is just $N-1$. If $f_\mathcal{P}(w)$ is the number of eligible choices for a second node if $w$ is chosen first, then the ratio is $\frac{\frac{1}{f_{\mathcal{P}'}(u)}+\frac{1}{f_{\mathcal{P}'}(v)}}{\frac{1}{f_\mathcal{P}(u)}+\frac{1}{f_\mathcal{P}(v)}}$.
 
\paragraph{Type B subtree slide} This time, $u$ is a random node whose parent exists and is in a different element of $\mathcal{P}$ to itself. This implies that $uP$ is in the same element as either $uS$ or $uG$ (if the latter exists) because otherwise $uP$ would not be in a partition element containing a tip. The operator performs the standard subtree slide move \cite{hohna_clock-constrained_2008} on $u$, inserting $uP$ as the parent of another node $v$ and (if $v$ was not the root node), the child of $vP$.  $\mathcal{P}$ is changed to a new partition $\mathcal{P}'$ as follows: if $vP$ does not exist or $v$ and $vP$ are in the same element of $\mathcal{P}$, $uP$ is moved to the element containing $v$. Otherwise, it is moved to either the element containing $v$ or that containing $vP$ with equal probability. This reallocation is enough to ensure that $\mathcal{P}'$ obeys connectedness rules. The effect on the transmission tree is that $\delta_\mathcal{P}(u)$ is moved to become a child of either $\delta_\mathcal{P}(v)$ or $\delta_\mathcal{P}(vP)$. If $\delta_\mathcal{P}(uS)\neq\delta_\mathcal{P}(uG)$ then $\delta_\mathcal{P}(uS)$ was the child of $\delta_\mathcal{P}(uG)$ before the move and remains so.

Noting that there are always $N-1$ choices for $u$, the Hastings ratio is the same as the standard subtree slide move, except that the denominator is multiplied by $\frac{1}{2}$ if $vP$ exists and $v$ and $vP$ are not in the same element of $\mathcal{P}$, and the numerator is multiplied by $\frac{1}{2}$ if $uG$ exists and $uG$ and $uS$ are not in the same element of $\mathcal{P}$.

\paragraph{Type B Wilson-Balding move}

In a similar way, $u$ is randomly picked from the set of nodes whose parents exist and are in different subtrees to themselves, and the standard Wilson-Balding move is performed on it, inserting $uP$ as a parent of another node $v$ and a child of its parent if that exists. The reassignment of $uP$ to a new subtree is performed in the same was as for type B subtree slide, and the adjustment to the Hastings ratio is identical. The effect on the transmission tree is also the same.

\subsection{Irreducibility of the chain}

Suppose $\mathcal{P}$ is a partition of a phylogeny $\mathcal{T}$ with root node $r$. First, notice the following about the infection branch operator described above:

\begin{itemize}
\item{For any $a_i\in\mathbf{A}$, if $\delta_\mathcal{P}(r)\neq a_i$, a series of downward moves, starting with one on $\epsilon_\mathcal{P}(a_i)$, eventually results in a new partition $\mathcal{P}'$ which has $\delta_{\mathcal{P}'}(r)=a_i$.}
\item{If $\delta_\mathcal{P}(r)=a_i$, a series of upward moves on $\epsilon_\mathcal{P}(a_j)$ for all $a_j\neq a_i$ will eventually give a partition $\mathcal{P}'$ in which $\delta_{\mathcal{P}'}(u)=a_i$ for all internal nodes $u$ of $\mathcal{T}$. As all such moves are reversible, we can get from $\mathcal{P}'$ to any partition $\mathcal{P}''$ that has $\delta_{\mathcal{P}''}(r)=a_i$.}
\end{itemize}

The above demonstrates that a MCMC chain made up of these moves on the space of partitions of a single phylogeny is irreducible. If $\mathcal{P}$ and $\mathcal{P}'$ are two partitions such that $\delta_\mathcal{P}(r)=\delta_{\mathcal{P}'}(r)$ then there is a series of moves taking $\mathcal{P}$ to $\mathcal{P}'$, and if $\delta_\mathcal{P}(r)\neq\delta_{\mathcal{P}'}(r)$ then there is a series of moves taking $\mathcal{P}$ to a partition $\mathcal{P}''$ that has $\delta_{\mathcal{P}''}(r)=\delta_{\mathcal{P}'}(r)$ and then a series of moves taking $\mathcal{P}''$ to $\mathcal{P}'$.

To extend this to a variable phylogenetic tree, we use the fact that in the space of standard, unpartitioned phylogenies, the Wilson-Balding move on its own is sufficient for irreducibility \cite{drummond_estimating_2002}. Suppose $\mathcal{P}$ is a partition of $\mathcal{T}$ such that $\delta_{\mathcal{P}'}(u)=\delta_\mathcal{P}(r)$ for all internal nodes $u$ of $\mathcal{T}$. Now every node of $\mathcal{T}$ is eligible to be the first node chosen by the type A Wilson-Balding move, as is true with the standard Wilson-Balding move on an unpartitioned tree, and subsequently, the set of nodes that is eligible to the the second node chosen is the same for both moves too. In addition, after this move, the new tree is still partitioned such that all internal nodes are in the same element of the partition. As a result, every move on an unpartitioned phylogeny that can be made by the standard move is also possible on the space of partitioned phylogenies with all internal nodes in the same partition element as a unique tip. Hence the chain is irreducible under the type A move when restricted to phylogenies with partitions of this type, and we have already shown that the infection branch operator is sufficient to move from a partition of this type to any other partition of the same tree. This is sufficient to establish irreducibility on the entire space of partitioned phylogenies using just these two moves.

\section{Bayesian decomposition}

Having established the correspondence between partitioned phylogenetic trees and transmission trees, we now show how the likelihood of such a partitioned phylogeny can be calculated given models of between-host transmission dynamics, of the duration of the infection within each host, of the population dynamics of the ``agents'' (which can be taken to be pathogens or infected individuals) within each host, and of sequence evolution. 

In contrast to the previous work of Didelot et al. \cite{didelot_bayesian_2014}, whose underlying model of transmission was a compartmental SIR model, we use an individual-based model similar to those employed in previous work on agricultural outbreaks \cite{cottam_integrating_2008, morelli_bayesian_2012, ypma_unravelling_2011, ypma_relating_2013}. This much more readily allows for the accommodation of host heterogeneity, and makes no assumption of random mixing. Instead, the force of infection of a host $a_i$ on another $a_j$ is given by a basic transmission rate $\beta$ multiplied by a positive real number $d(h_1,h_2)$ from a function $d:\mathbf{A}\times\mathbf{A}\rightarrow[0,\infty)$ describing some relationship between $a_i$ and $a_j$. Possible choices for $d$ are a spatial kernel function, a network metric, or a function modifying $\beta$ based on shared membership in some class of host.

As in previous work \cite{ypma_relating_2013, didelot_bayesian_2014} we take the model of the dynamics of the ``agents'' to be a coalescent process amongst lineages in a freely-mixing population within each host. If the hosts are single organisms, the agents will naturally be individual pathogens. If, on the other hand, the hosts are infected locations, they could instead be considered to be infected organisms. In either case, only a miniscule proportion of the total agent population are represented by lineages in the tree, and the assumption of a low sampling fraction required for use of the coalescent process is satisfied.

We use the following notation:

\begin{itemize}
\item{The sequence data, $D$}
\item{The phylogenetic tree, $\mathcal{T}$}
\item{The transmission tree structure, $\mathcal{N}$}
\item{The set $\mathbf{T}^{\mathrm{inf}}$ of times of infection of each host}
\item{The times of sampling $\mathbf{T}^{\mathrm{exam}}$ of the sequence from each host}
\item{The times of becoming noninfectious $\mathbf{T}^{\mathrm{end}}$ of each host.}
\item{Data $L$ describing the relationship between hosts that is used to define the function $d$ (for example, spatial locations).}
\item{The basic transmission rate $\beta$.}
\item{The parameters $\phi$ of the distance function $d$.}
\item{The parameters $\psi$ of the population dynamics of the agents within each host.}
\item{The parameters $\omega$ of the nucleotide substitution model and molecular clock.}
\end{itemize}

We condition on $\mathbf{T}^{\mathrm{exam}}$, $\mathbf{T}^{\mathrm{end}}$, and $L$. We assume that $\mathbf{T}^{\mathrm{exam}}$ and $\mathbf{T}^{\mathrm{end}}$ are not contradictory; no sample was taken after a host became noninfectious. $\mathbf{T}^{\mathrm{end}}$ can be the same set of times as $\mathbf{T}^{\mathrm{exam}}$, or a separate set of later times. If any or all hosts are known to have remained infectious indefinitely, their values of $\mathbf{T}^{\mathrm{end}}$ can be set to the time at which the last sample was taken.

The posterior probability we are interested in calculating is $p(\mathcal{T}, \mathcal{N}, \mathbf{T}^{\mathrm{inf}}, \beta, \phi, \psi, \omega | D, \mathbf{T}^{\mathrm{exam}}, \mathbf{T}^{\mathrm{end}}, L)$. By Bayes' Theorem this is equal to:
\begin{eqnarray*}
\frac{p(D|\mathcal{T}, \mathcal{N}, \mathbf{T}^{\mathrm{inf}}, \beta, \phi, \psi, \omega, \mathbf{T}^{\mathrm{exam}}, \mathbf{T}^{\mathrm{end}}, L)p(\mathcal{T}, \mathcal{N}, \mathbf{T}^{\mathrm{inf}}, \beta, \phi, \psi, \omega, \mathbf{T}^{\mathrm{exam}}, \mathbf{T}^{\mathrm{end}}, L)}{p(D|\mathbf{T}^{\mathrm{exam}}, \mathbf{T}^{\mathrm{end}}, L)}
\end{eqnarray*}

As usual, we need not calculate the denominator if we are uninterested in model comparison as it does not vary. We assume that mutations occur neutrally over the the phylogenetic tree in a process that ignores the host structure, so $D$ depends only on $\mathcal{T}$ and $\omega$ and the likelihood reduces to $p(D|\mathcal{T},\omega)$, which can be calculated using the Felsenstein pruning algorithm and a molecular clock model in the normal way \cite{felsenstein_evolutionary_1981, drummond_estimating_2002, drummond_relaxed_2006}. It remains to calculate the prior probability $p(\mathcal{T}, \mathcal{N}, \mathbf{T}^{\mathrm{inf}}, \phi, \psi, \omega, \mathbf{T}^{\mathrm{exam}}, \mathbf{T}^{\mathrm{end}}, L)$. The full decomposition is as follows:
\begin{eqnarray*}
p(\mathcal{T}, \mathcal{N}, \mathbf{T}^{\mathrm{inf}}, \phi, \psi, \omega, \mathbf{T}^{\mathrm{exam}}, \mathbf{T}^{\mathrm{end}}, L) &=& p(\beta|\mathcal{T}, \mathcal{N}, \mathbf{T}^{\mathrm{inf}}, \phi, \psi, \omega, \mathbf{T}^{\mathrm{exam}}, \mathbf{T}^{\mathrm{end}}, L)\\
&& \times p(\mathcal{T} | \mathcal{N}, \mathbf{T}^{\mathrm{inf}}, \phi, \psi, \omega, \mathbf{T}^{\mathrm{exam}}, \mathbf{T}^{\mathrm{end}}, L)\\
&& \times p(\mathcal{N} | \mathbf{T}^{\mathrm{inf}}, \phi, \psi, \omega, \mathbf{T}^{\mathrm{exam}}, \mathbf{T}^{\mathrm{end}}, L)\\
&& \times p(\mathbf{T}^{\mathrm{inf}} | \phi, \psi, \omega, \mathbf{T}^{\mathrm{exam}}, \mathbf{T}^{\mathrm{end}}, L)\\
&& \times p(\phi, \psi, \omega | \mathbf{T}^{\mathrm{exam}}, \mathbf{T}^{\mathrm{end}}, L)
\end{eqnarray*}

The following assumptions of independence are then made:

\begin{itemize}
\item{All parameters in the decomposition are independent of $\omega$.}
\item{$\beta$, the base transmission rate, is (at least) conditionally independent of $\mathcal{T}$ and $\psi$ given $\phi$, $\mathcal{N}$, $\mathbf{T}^{\mathrm{exam}}$, $\mathbf{T}^{\mathrm{inf}}$, $\mathbf{T}^{\mathrm{end}}$, and $L$. This is intuitive given that the latter set of parameters completely describe the epidemic and the distance-based modification of $\beta$.}
\item{$\mathcal{T}$, the phylogenetic tree, is (at least) conditionally independent of $\phi$, $\mathbf{T}^{\mathrm{end}}$, and $L$  given $\psi$, $\mathcal{N}$, $\mathbf{T}^{\mathrm{inf}}$, and $\mathbf{T}^{\mathrm{exam}}$.}
\item{$\mathcal{N}$, the transmission tree structure, is (at least) conditionally independent of $\mathbf{T}^{\mathrm{exam}}$ and $\psi$ given $\phi$, $\mathbf{T}^{\mathrm{inf}}$, $\mathbf{T}^{\mathrm{end}}$ and $L$.}
\item{$\mathbf{T}^{\mathrm{inf}}$, the times of infection, is (at least) conditionally independent of $\phi$, $\mathbf{T}^{\mathrm{exam}}$, $\psi$ and $L$ given $\mathbf{T}^{\mathrm{end}}$.}
\item{$\phi$, $\psi$ and $\omega$ are independent of $\mathbf{T}^{\mathrm{inf}}$, $\mathbf{T}^{\mathrm{end}}$, $L$, and each other.}
\end{itemize}

The decomposition then reduces to:

\begin{eqnarray*}
p(\mathcal{T}, \mathcal{N}, \mathbf{T}^{\mathrm{inf}}, \phi, \psi, \omega, \mathbf{T}^{\mathrm{exam}}, \mathbf{T}^{\mathrm{end}}, L) &=& p(\beta|\mathcal{N}, \mathbf{T}^{\mathrm{inf}}, \phi, \mathbf{T}^{\mathrm{exam}}, \mathbf{T}^{\mathrm{end}}, L)\\
&& \times p(\mathcal{T} | \mathcal{N}, \mathbf{T}^{\mathrm{inf}}, \psi, \mathbf{T}^{\mathrm{exam}})\\
&& \times p(\mathcal{N} | \mathbf{T}^{\mathrm{inf}}, \phi, \mathbf{T}^{\mathrm{end}}, L)\\
&& \times p(\mathbf{T}^{\mathrm{inf}} |  \mathbf{T}^{\mathrm{end}})\\
&& \times p(\phi)p(\psi)p(\omega)
\end{eqnarray*}

For calculation of $p(\beta|\mathcal{N}, \mathbf{T}^{\mathrm{inf}}, \phi, \mathbf{T}^{\mathrm{exam}}, \mathbf{T}^{\mathrm{end}}, L)$, we use Bayes' Theorem again: 

\begin{eqnarray*}
p(\beta|\mathcal{N}, \mathbf{T}^{\mathrm{inf}}, \phi, \mathbf{T}^{\mathrm{exam}}, \mathbf{T}^{\mathrm{end}}, L)&=&\frac{p(\mathcal{N}, \mathbf{T}^{\mathrm{inf}}|\beta, \phi, \mathbf{T}^{\mathrm{exam}}, \mathbf{T}^{\mathrm{end}}, L)p(\beta |\phi, \mathbf{T}^{\mathrm{exam}}, \mathbf{T}^{\mathrm{end}}, L)}{p(\mathcal{N},\mathbf{T}^{\mathrm{inf}}|\phi,\mathbf{T}^{\mathrm{exam}}, \mathbf{T}^{\mathrm{end}}, L)}
\end{eqnarray*}

The denominator $p(\mathcal{N},\mathbf{T}^{\mathrm{inf}}|\phi,\mathbf{T}^{\mathrm{exam}}, \mathbf{T}^{\mathrm{end}}, L)$ can be evaluated as
\begin{eqnarray*}
\int_\beta p(\mathcal{N}, \mathbf{T}^{\mathrm{inf}}|\beta, \phi, \mathbf{T}^{\mathrm{exam}}, \mathbf{T}^{\mathrm{end}}, L)p(\beta |\phi, \mathbf{T}^{\mathrm{exam}}, \mathbf{T}^{\mathrm{end}}, L)d\beta
\end{eqnarray*}

by the law of total probability. This will not in general have a closed form solution and we use numerical integration to estimate it. The term $p(\beta |\phi,\mathbf{T}^{\mathrm{exam}}, \mathbf{T}^{\mathrm{end}}, L)$ is our prior belief in the value of $\beta$ given $\phi$ and the background information; in the absence of other information we take $\beta$ to be independent of these and simply give it any prior distribution $p(\beta)$ that we please.

It remains to calculate $p(\mathcal{N}, \mathbf{T}^{\mathrm{inf}}|\beta, \phi, \mathbf{T}^{\mathrm{exam}}, \mathbf{T}^{\mathrm{end}}, L)$.  The calculation here is along the same lines of that introduced by Gibson and Austin \cite{gibson_fitting_1996}, but heavily modified. Given a particular set of $\mathbf{T}^{\mathrm{inf}}$, we reorder the indexes of $\mathbf{A}$ to be in increasing order of infection. As before, the infection time of $a_i$ is $t^{\mathrm{inf}}_i$. The probability that $a_1$ was infected at time $t^{\mathrm{inf}}_1$ given that it was first in the epidemic is effectively unknowable and we set it to 1. For notational simplicity now treating $\mathcal{N}$ as a map from the index of a case to the index of its infector, we need the probability that $a_{\mathcal{N}(i)}$ infected $a_i$ at $t^{\mathrm{inf}}_i$, which is made up of:

\begin{itemize}
\item{The probability that $a_{\mathcal{N}(i)}$ infected $a_i$ at $t^{\mathrm{inf}}_i$, but not before:
\begin{eqnarray*}
\beta d(a_i,a_{\mathcal{N}(i)})\times\mathrm{exp}\left(-\beta d(a_i,a_{\mathcal{N}(i)})(t^{\mathrm{inf}}_i-t^{\mathrm{inf}}_{\mathcal{N}(i)})\right)
\end{eqnarray*}
}
\item{The probability that no other host infected $a_i$ before $t^{\mathrm{inf}}_i$. As we assume no reinfection, infection events that would occur after this time are ignored. Noting that the last possible time that an $a_j$ could have infected $a_i$ for this to be true is the smaller of $t^{\mathrm{inf}}_i$ and the end of of $a_j$'s infectiousness, $t^{\mathrm{end}}_j$, this is given by:
\begin{eqnarray*}
\prod_{j\in\{1,\ldots,i-1\}\setminus \mathcal{N}(i)}\mathrm{exp}\left({-\beta d(a_i,a_j)(\textrm{min}\{t^{\mathrm{inf}}_i,t^{\mathrm{end}}_j\}-t^{\mathrm{inf}}_j)}\right)
\end{eqnarray*}
}
\item{As we are conditioning on $\mathbf{T}^{\mathrm{exam}}$ and $\mathbf{T}^{\mathrm{end}}$, we implicitly assume that each $a_i$ was, in fact, infected and was infected before its time of sampling $t^{\mathrm{exam}}_i$. As a result, we need to normalise by the probability that an infection did happen before this date, which is one minus the probability that none did. The last possible time that an $a_j$ could have infected $a_i$ at all is the smaller of $t^{\mathrm{exam}}_i$ and $t^{\mathrm{end}}_j$, so this expression is:
\begin{eqnarray*}
1-\prod_{\substack{j\in\{1,\ldots,N\}\\t^{\mathrm{inf}}_j<t^{\mathrm{exam}}_i}} \mathrm{exp}\left({-\beta d(a_i,a_j)(\textrm{min}\{t^{\mathrm{exam}}_i,t^{\mathrm{end}}_j\}-t^{\mathrm{inf}}_j)}\right)
\end{eqnarray*}
}
\end{itemize}

Thus the full expression for $p(\mathcal{N}, \mathbf{T}^{\mathrm{inf}}|\beta, \phi, \mathbf{T}^{\mathrm{exam}}, \mathbf{T}^{\mathrm{end}}, L)$ is:
\begin{eqnarray*}
\mathlarger{\prod}_{i\in\{2,\ldots,N\}}\left(\frac{\beta d(a_i,a_{\mathcal{N}(i)})\prod_{j\in\{1,\ldots,i-1\}}\mathrm{exp}\left(-\beta d(a_i,a_j)(\textrm{min}\{t^{\mathrm{inf}}_i,t^{\mathrm{end}}_j\}-t^{\mathrm{inf}}_j)\right)}{1-\prod_{\substack{j\in\{1,\ldots,N\}\\t^{\mathrm{inf}}_j<t^{\mathrm{exam}}_i}} \textrm{exp}\left({-\beta d(a_i,a_j)(\textrm{min}\{t^{\mathrm{exam}}_i,t^{\mathrm{end}}_j\}-t^{\mathrm{inf}}_j)}\right)}\right)
\end{eqnarray*}

If one of the $d(a_i,a_j)$ terms in this expression is zero, then the whole thing is zero, representing an impossible transmission history. Otherwise this is undefined for $\beta=0$ but exists and is positive for all other $\beta\in(0,\infty)$, as the denominator is always greater than zero. Let this expression, as a function of $\beta$ alone with all other variables constant, be $I(\beta)$. The integral $\int_0^\infty I(\beta)p(\beta)d\beta$, where $p(\beta)$ is the prior probability of $\beta$, can estimated by numerical methods if we show that it is in fact finite.

\begin{proposition}
Let $a,b\in(0,\infty)$. The improper integral $\int_a^b I(\beta)d\beta$ converges as $a\to0$ and $b\to\infty$.
\end{proposition}

\begin{proof}
For the lower limit, we use:

\begin{lemma}\label{lop}
Suppose $A$ and $B$ are positive real numbers. Then:
\begin{eqnarray*}
\lim_{x\to0} \frac{xe^{-Ax}}{1-e^{-Bx}}=\frac{1}{B}
\end{eqnarray*}
\end{lemma}

\begin{proof}
Let $f(x)=xe^{-Ax}$ and $g(x)=1-e^{-Bx}$. Then $f'(x)=(1-Ax)e^{-Ax}$ and $g'(x)=Be^{-Bx}$. Hence:

\begin{eqnarray*}
\frac{f'(x)}{g'(x)}&=&\frac{(1-Ax)e^{-Ax}}{Be^{-Bx}}\\
&=& \frac{1-Ax}{B}e^{(B-A)x}\\
\end{eqnarray*}
This shows that $\lim_{x\to0}f'(x)/g'(x)=1/B$ and the result follows by l'H\^{o}pital's rule.

\end{proof}

Lemma~\ref{lop} shows that each individual term in the product that makes up $I(\beta)$ does not have 0 as an asymptote, hence they are all bounded on $(0,a]$ as they clearly have no others. Hence, on this interval, $I(\beta)$, as the product of bounded functions, is bounded and the integral converges.

For the upper limit, we can write:

\begin{eqnarray*}
I(\beta) = \frac{A\beta^{n-1}\textrm{exp }(-B\beta)}{\prod_{i=2}^N(1-\textrm{exp }(-C_i\beta))}\\
\end{eqnarray*}

where $A$, $B$ and each $C_i$ is a positive real number. If we let $J(\beta)=A\beta^{n-1}\textrm{exp }(-B\beta)$ then $J(\beta)/I(\beta)=\prod_{i=2}^N(1-\textrm{exp }(-C_i\beta))$ whose limit as $\beta\to\infty$ is 1. The limit comparison test then says that $\int_a^bI(\beta)d\beta$ converges as $b\to\infty$ if and only if $\int_a^bJ(\beta)d\beta$ does. Recursive integration by parts gives:

\begin{eqnarray*}
\int_a^bJ(\beta)d\beta  = \left[ \frac{A}{B}\left(\sum_{k=0}^{n-1}\left(\frac{-1}{B}\right)^{n-1-k}\beta^{k}\right)\textrm{exp }(-B\beta)\right]_a^b\\
\end{eqnarray*}

$\int_a^\infty J(\beta)d\beta = \lim_{b\to\infty}\int_a^b J(\beta)d\beta$, and $\int_a^\infty J(\beta)d\beta$ can thus be expressed as a constant expression involving $a$, plus the sum of $n$ limits of the form $\lim_{\beta\to\infty}D\beta^k\textrm{exp }(-B\beta)$ where $D$ is a constant and $k\in\mathbb{N}$. It is a standard result that each of these is 0. Hence $J(\beta)$ converges and so does $I(\beta)$.

\end{proof}

\begin{corollary}
Let $a,b\in(0,\infty)$. If $p(\beta)$ is a proper prior distribution whose support is a subset of $(0,\infty)$, the improper integral $\int_a^b I(\beta)p(\beta)d\beta$ converges as $a\to0$ and $b\to\infty$.
\end{corollary}

\begin{proof}
If $p(\beta)$ has finite support, then $I(\beta)p(\beta)$ is bounded on a finite interval and zero elsewhere, and the intergral of such a function must converge. If not, then use of the limit comparison test with numerator $I(\beta)p(\beta)$ and denominator $I(\beta)$ gives that, because $\lim_{\beta\to\infty}p(\beta)=0$, $\int_a^b I(\beta)p(\beta)d\beta$ converges if $\int_a^b I(\beta)d\beta$ does, and we know this to be true.
\end{proof}

\begin{remark}
Notice that if $p(\beta)$ is, for example a uniform infinite improper prior or a gamma distribution, $I(\beta)p(\beta)$ takes the form $D\beta^k\textrm{exp }(-E\beta)f(\beta)$ for a function $f$ where $D$ and $E$ are positive constants and $k>1$. Generalised Gauss-Laguerre quadrature is therefore a natural choice for the estimation of $\int_0^\infty I(\beta)p(\beta)d\beta$ for such a $p(\beta)$. 
\end{remark}

Next, we need to calculate $p(\mathcal{T} | \mathcal{N}, \mathbf{T}^{\mathrm{inf}}, \psi, \mathbf{T}^{\mathrm{exam}})$. We extend the procedure outlined by Didelot et al\cite{didelot_bayesian_2014} to allow for the use of any of the standard models of deterministic population growth, and the possibility of host heterogeneity. The latter is accomplished by dividing the set of hosts into categories and assigning a separate demographic model to all the hosts in each one. Categories can be assigned from known epidemiological data about the hosts; for example, in a livestock disease outbreak, they may reflect the size of farm. Formally, let $\mathbf{C}^{\textrm{coal}}$, a finite set of size $p$, be the set of categories, and $cc:\{1,\ldots,N\}\to\mathbf{C}^{\textrm{coal}}$ the map assigning them to the index of each host in $A$. If it is not desired to accommodate heterogeneity in this way, $p$ can be 1. Every element $\mathbf{c}\in\mathbf{C}^{\textrm{coal}}$ corresponds to a separate demographic function $N_\mathbf{c}:\mathbb{R}\rightarrow[0,\infty)$ with parameters $\psi_\mathbf{c}$ where $N_\mathbf{c}(t)$ is the product of the effective population size and the generation time at time $t$.

Given a host $a_i\in\mathbf{A}$ which is infected at time $t^{\mathrm{inf}}_i$, sampled at time $t^{\mathrm{exam}}_i$ and ceases to be infectious at time $t^{\mathrm{end}}_i$, and has $n$ children $a_{o(1)},\ldots,a_{o(n)}$ (for some permutation $o$ of $\{1,\ldots,N\}$) infected at times $t^{\mathrm{inf}}_{o(1)},\ldots,t^{\mathrm{inf}}_{o(n)}$, suppose $\mathcal{S}_i$ is a phylogenetic tree that describes the part of the outbreak that took place within $a_i$. It has has $n+1$ tips, one for each infection event and one for its own sampling event. If $m=\textrm{max }\{t^{\mathrm{inf}}_{o(1)},\ldots,t^{\mathrm{inf}}_{o(n)},t^{\mathrm{exam}}_{i}\}$, the height (in the tree $\mathcal{S}_i$) $h_i(r)$ of its root node $r$ is less than $m-t^{\mathrm{inf}}_i$ and we can give it a root branch of length $m-h_i(r)-t^{\mathrm{inf}}_i$. If we have a $\mathcal{S}_i$ for each $h$, and we know $\mathcal{N}$, we can build a phylogenetic tree for the entire epidemic by attaching the root node of each $\mathcal{S}_i$ to the tip of $\mathcal{S}_{\mathcal{N}(i)}$ that corresponds to the infection of $a_i$, by a branch with length equal to the root branch length of $\mathcal{S}_i$. If $\mathcal{T}$ cannot be built up from $\mathcal{S}_i$s in this way, $p(\mathcal{T} | \mathcal{N}, \mathbf{T}^{\mathrm{inf}}, \psi, \mathbf{T}^{\mathrm{exam}})=0$. Otherwise, we calculate it as:
\begin{eqnarray*}
p(\mathcal{T} | \mathcal{N}, \mathbf{T}^{\mathrm{inf}}, \psi, \mathbf{T}^{\mathrm{exam}})= \prod_{i\in\{1,\ldots,N\}}p(\mathcal{S}_{i}|\psi_{cc(i)})
\end{eqnarray*}
In the standard coalescent model \cite{slatkin_pairwise_1991}, the probability density function for the for the time $t$ of the first coalescence of $K\geq2$ lineages after $t_0$ where the demographic function is $N_\mathbf{c}$ is given by:

\begin{eqnarray*}
p(t)&=&\frac{K(K-1)}{2N_\mathbf{c}(t)}\textrm{ exp}\left(-\int_{t_0}^t\frac{K(K-1)}{2N_\mathbf{c}(s)}ds\right)\\
\end{eqnarray*}
and as usual, if we know the two specific lineages that converged, the $K(K-1)/2$ cancels.

The cumulative density function of this is:
\begin{eqnarray*}
P(t)&=&\int_{t_0}^t\frac{K(K-1)}{2N_\mathbf{c}(r)}\textrm{ exp}\left(-\int_{t_0}^r\frac{K(K-1)}{2N_\mathbf{c}(s)}ds\right)dr\\
&=&1-\textrm{exp}\left(-\int_{t_0}^t\frac{K(K-1)}{2N_\mathbf{c}(s)}ds\right)
\end{eqnarray*}

and the probability that there were no coalescences between $t_0$ and $t$ is 1 minus this.

As Didelot et al. \cite{didelot_bayesian_2014} note, this is not quite sufficient for our purposes because we have a maximum height for the last coalescence. If this is $t_{\mathrm{max}}$, the normalised probability distribution for the time of first coalescence is:

\begin{eqnarray*}
p(t|T) =  \begin{cases} \frac{\frac{K(K-1)}{2N_\mathbf{c}(t)}\textrm{exp}\left(-\int_{t_0}^t\frac{K(K-1)}{2N_\mathbf{c}(s)}ds\right)}{1-\textrm{exp}\left(-\int_{t_0}^{t_{\mathrm{max}}}\frac{K(K-1)}{2N_\mathbf{c}(s)}ds\right)} & t_0\leq t<t_{\mathrm{max}}\\
0 & \mbox{otherwise}
\end{cases}
\end{eqnarray*}

This is the probability of an interval in $\mathcal{S}_h$ ending in a coalescent event. The probability of an interval ending in a transmission or sampling event is the probability that no events occur in the interval, which is one minus the cumulative distribution function of the above, $P(t|T)$:

\begin{eqnarray*}
1-P(t|T)&=&1-\frac{1-\textrm{exp}\left(-\int_{t_0}^t\frac{K(K-1)}{2N_\mathbf{c}(s)}ds\right)}{1-\textrm{exp}\left(-\int_{t_0}^{t_{\mathrm{max}}}\frac{K(K-1)}{2N_\mathbf{c}(s)}ds\right)}\\
&=&\frac{\textrm{exp}\left(-\int_{t_0}^t\frac{K(K-1)}{2N_\mathbf{c}(s)}ds\right)-\textrm{exp}\left(-\int_{t_0}^{t_{\mathrm{max}}}\frac{K(K-1)}{2N_\mathbf{c}(s)}ds\right)}{1-\textrm{exp}\left(-\int_{t_0}^{t_{\mathrm{max}}}\frac{K(K-1)}{2N_\mathbf{c}(s)}ds\right)}\\
\end{eqnarray*}

Note that while in the case of no maximum root height, the formula happens to work for $K=1$, here it does not as the denominator is 0, and we instead set the probability of any coalescent interval with one lineage to 1. In particular, if $a_i$ has no children then $p(\mathcal{S}_{i}|\psi_{cc(i)})=1$.

If $t_{\mathrm{max}}=m-t^{\mathrm{inf}}_i$, these formulae can be used to calculate $p(\mathcal{S}_i|\psi_{cc(i)})$ for every $\mathcal{S}_{i}$ in the established way for a tree with temporally offset tips \cite{drummond_estimating_2002}, and the product of these is the full probability of the complete phylogeny. It is most intuitive to standardise the timescale within each $\mathcal{S}_{i}$ such that the effective population size at the point of the infection (the maximum root height) is the same across all hosts. As a result, we depart from the normal convention of making height 0 the time of the last tip (which will occur at a different point in the course of infection in different hosts), and instead put it at the point of infection, with all later events occurring at negative heights.

The choice of each demographic function $N_\mathbf{c}$ is wide. For an epidemic situation, exponential or logistic growth \cite{slatkin_pairwise_1991, pybus_epidemic_2001} would be most appropriate. Different categories $\mathbf{c}\in\mathbf{C}$ may be assigned the same family of demographic model but a different set of parameters $\psi_\mathbf{c}$. As our method is integrated within BEAST, any of the functions already implemented in that package can be used without additional programming work.

The next term in the decomposition is $p(\mathcal{N} | \mathbf{T}^{\mathrm{inf}}, \phi, \mathbf{T}^{\mathrm{end}}, L)$. The instantaneous probability that host $a_i$ was infected by host $a_{\mathcal{N}(i)}$ at time $t^{\mathrm{inf}}_i$ is $\beta d(a_i,a_{\mathcal{N}(i)})$, and if we condition on the fact that $a_i$ was indeed infected by \emph{some} host at $t^{\mathrm{inf}}_i$ then we normalise by the sum $\sum_{a_j\in\mathbf{A}_i}\beta d(a_i,a_j)$ where $\mathbf{A}_i$ is the subset of $\mathbf{A}$ whose elements have infection times before $t^{\mathrm{inf}}_i$ and noninfectiousness times after it. In this normalisation the $\beta$s cancel, leaving an expression solely in terms of the distance function. The probability of the infection of the first host $a_1$ is once again set to 1. The expression is:

\begin{eqnarray*}
p(\mathcal{N} | \mathbf{T}^{\mathrm{inf}}, \phi, \mathbf{T}^{\mathrm{end}}, L) = \prod_{a_i\in\mathbf{A}\setminus a_1}\frac{d(a_i,a_{\mathcal{N}(i)})}{\sum_{a_j\in\mathbf{A}_i}d(a_i,a_j)}
\end{eqnarray*}

There are many possible choices for the function $d$. If we assume no spatial structure or heterogeneity then we can just take $d(a_i,a_j)=1$ for all $a_i,a_j\in\mathbf{A}$. Otherwise, it can be based on Euclidean distance, or on a network metric. It can also be used to state prior information about the transmission tree structure; if it is known \emph{a priori} that $a_i$ did not infect $a_j$, then $d(a_i,a_j)$ can be set to zero. While we have assumed it up to this point, there is also no requirement that $d$ be symmetric.

The calculation of $p(\mathbf{T}^{\mathrm{inf}} | \mathbf{T}^{\mathrm{end}})$, the probability of the times of infection, can be handled in a number of ways. It is effectively the calculation of the probability of the time from infection to noninfectiousness, $t^{\mathrm{end}}_i - t^{\mathrm{inf}}_i$, of each host $a_i$. Previous work on foot-and-mouth disease virus \cite{cottam_integrating_2008, morelli_bayesian_2012} has used clinical data to estimate times of infection, and if this kind of information is available, it can be used to determine a separate prior distribution for each $t^{\mathrm{end}}_i - t^{\mathrm{inf}}_i$. If we cannot use information of this type, we take a similar approach to that in the coalescent calculations above and assign each host $a_i$ to a category $ic(a_i)$ from a finite set $\mathbf{C}^{\mathrm{inf}}$ of size $q$. This again allows us to accommodate known host heterogeneity; for example in an agricultural outbreak it is likely that times from infection to noninfectiousness decrease as time goes by and control measures are brought to bear. Once again, if we do not want to incorporate such heterogeneity we can set $q=1$. If the infectious period of the disease is well understood, we can assign a single prior distribution for $t^{\mathrm{end}}_i - t^{\mathrm{inf}}_i$ for all hosts in each category. 

It may be, however, that we want to estimate the distribution of infectious periods from the genetic data. In this case we take each $t^{\mathrm{end}}_i - t^{\mathrm{inf}}_i$ within a category as a draw from a probability distribution with unknown parameters, and then put hyperpriors on those parameters. A noninformative option is to regard each as a draw from an unknown normal distribution whose mean we are uninterested in (as we can just as well calculate the mean of the sampled values of each infectious period post-hoc) and use the Jeffreys prior on its standard deviation, such that $p(\mathbf{T}^{\mathrm{inf}} | \mathbf{T}^{\mathrm{end}})$ is proportional to the reciprocal of the standard deviation of all the infectious periods in the category. This can obviously be done on the logarithm of the infectious periods instead, if we prefer the assumption that they are lognormally distributed. Alternatively, we can use an informative prior. To avoid having to use MCMC to estimate both the parameters $\chi$ of a probability distribution $D$ and a series of draws from that distribution, we integrate out the actual values of $\chi$ by using $D$'s conjugate prior for them and then calculating the marginal likelihood of the infectious periods given the hyperpriors. Any continuous probability distribution with a prior whose marginal likelihood is analytically tractable can be considered. A normal distribution is not absolutely ideal as infectious periods are non-negative parameters, but it does have the useful property that its mean and variance are independent, unlike most other candidates for $D$ (such as lognormal, exponential or gamma). We suggest it still be considered as an option if infectious periods are expected to be sufficiently long, and their variance sufficiently small, that the probability density contained in the area less than 0 would negligible if a normal distribution were used.

Finally, all that remains is to place prior distributions on the parameters making up $\phi$, $\psi$, and $\omega$.

\subsection*{Latent periods}

The above formulation has taken the course of infection to follow a SIR structure; hosts are assumed to be infectious as soon as they are infected. It is straightforward to replace this with a SEIR structure instead. We add an extra set of parameters $\mathbf{T}^{\textrm{trans}}$ consisting of the time of infectiousness $t^{\textrm{trans}}_i$ of each host $a_i\in\mathbf{A}$. In the MCMC procedure these are calculated by adding a $r_{i}\in[0,1]$ such that, if $t^{\textrm{maxtrans}}_i=\textrm{min }(\{t^{\textrm{end}}_i\}\cup\{t^{\textrm{inf}}_j : \mathcal{N}(j)=i\})$ ($t^{\textrm{maxtrans}}_i$ being the upper bound on $t^{\textrm{trans}}_i$ determined by $\mathbf{T}^{\textrm{inf}}$ and $\mathbf{T}^{\textrm{end}}$), then $t^{\textrm{trans}}_i = t^{\textrm{inf}}_i + r_i(t^{\textrm{maxtrans}}_{i}-t^{\textrm{inf}})$.  (We assume that hosts are infectious at the time they cease to be infected, but not that they necessarily are at the time of sampling.) Simple MCMC moves on numerical parameters are then employed to sample values of each $r_{i}$. The phylogeny $\mathcal{T}$ is assumed to be conditionally independent of $\mathbf{T}^{\textrm{trans}}$ given $\mathbf{T}^{\textrm{inf}}$.

The decomposition becomes:

\begin{eqnarray*}
p(\mathcal{T}, \mathcal{N}, \mathbf{T}^{\mathrm{inf}}, \mathbf{T}^{\textrm{trans}}, \phi, \psi, \omega, \mathbf{T}^{\mathrm{exam}}, \mathbf{T}^{\mathrm{end}}, L) &=& p(\beta|\mathcal{N}, \mathbf{T}^{\mathrm{inf}}, \mathbf{T}^{\textrm{trans}}, \phi, \mathbf{T}^{\mathrm{exam}}, \mathbf{T}^{\mathrm{end}}, L)\\
&& \times p(\mathcal{T} | \mathcal{N}, \mathbf{T}^{\mathrm{inf}}, \psi, \mathbf{T}^{\mathrm{exam}})\\
&& \times p(\mathcal{N} | \mathbf{T}^{\mathrm{inf}}, \mathbf{T}^{\textrm{trans}}, \phi, \mathbf{T}^{\mathrm{end}}, L)\\
&& \times p(\mathbf{T}^{\mathrm{inf}}, \mathbf{T}^{\textrm{trans}} |  \mathbf{T}^{\mathrm{end}})\\
&& \times p(\phi)p(\psi)p(\omega)
\end{eqnarray*}

The only modifications to the existing procedure that are needed to calculate the first and third elements in this product involve accounting for the fact that infectious pressure is now only applied after the end of a host's latent period. The term $p(\mathbf{T}^{\mathrm{inf}}, \mathbf{T}^{\textrm{trans}} |  \mathbf{T}^{\mathrm{end}})$ is calculated by retaining the existing procedure for infectious periods and also picking a suitable distribution, or possibly set of distributions according to a set $\mathbf{C}^{\mathrm{lat}}$ of categories, for latent periods. The marginal likelihood of the set of infectious periods $\{t^{\mathrm{end}}_i - t^{\mathrm{trans}}_i|i\in\{1,\ldots,N\}\}$ and latent periods $\{t^{\mathrm{trans}}_i - t^{\mathrm{inf}}_i|i\in\{1,\ldots,N\}\}$ is then calculated as before. 

\section{Conclusions and future work}

The most obvious limitation to the method outlined here is the requirement that the tree contain a tip from every host involved in the outbreak. Note that this is not, in fact, a requirement that a sequence be taken from every host, as, if we have no sample from some but are aware of their existence, we can use epidemiological data for them along with a noninformative sequence (consisting entirely of the nucleotide code `N') and integrate over their unknown true sequences. Their existence would then still contribute to estimation of epidemiological parameters, and their estimated placement in the transmission tree would be informed by their geographical locations if those were included in the model. The performance of this procedure for varying numbers of unknown sequences warrants investigation in simulations. Nevertheless, this is a solution only where all unsampled hosts are actually known to investigators, which will not always be the case. Non-phylogenetic methods for transmission tree reconstruction using genetic data have started to consider the case of unsampled hosts \cite{jombart_bayesian_2014,mollentze_bayesian_2014} and work is needed to introduce this element to our framework. This could be done by introducing a variable number of partitions containing no tips, possibly in a reversable-jump MCMC framework, or by simply assigning tree nodes to ``nonsampled'' subtrees with no specific enumeration of how many extra hosts these represent.

Another enhancement, potentially useful in HIV studies where multiple samples are taken from the same patient over time, would be to relax the partition rules to allow each subtree to contain more than one tip. The adjustments to the method needed to accomplish this are likely considerably less onerous than allowing for unsampled hosts.

In conclusion, in this document we have outlined the framework for co-estimation of transmission trees as part of an analysis performed in one of the most widely-used software packages for phylogeny reconstruction. Results from analyses of both simulated and real data will follow. It is, as of the time of writing, implemented in current development builds of BEAST.

\section{Acknowledgements}

We would like to thank Trevor Bedford, Samantha Lycett and Melissa Ward for contributions to the development of this model. MH was supported by a PhD studentship from the Scottish Government-funded EPIC programme, and the research leading to these results has received funding from the European Union Seventh Framework Programme for research, technological development and demonstration under Grant Agreement no. 278433-PREDEMICS

\begin{figure*}
\centering
\includegraphics[width=16.0cm]{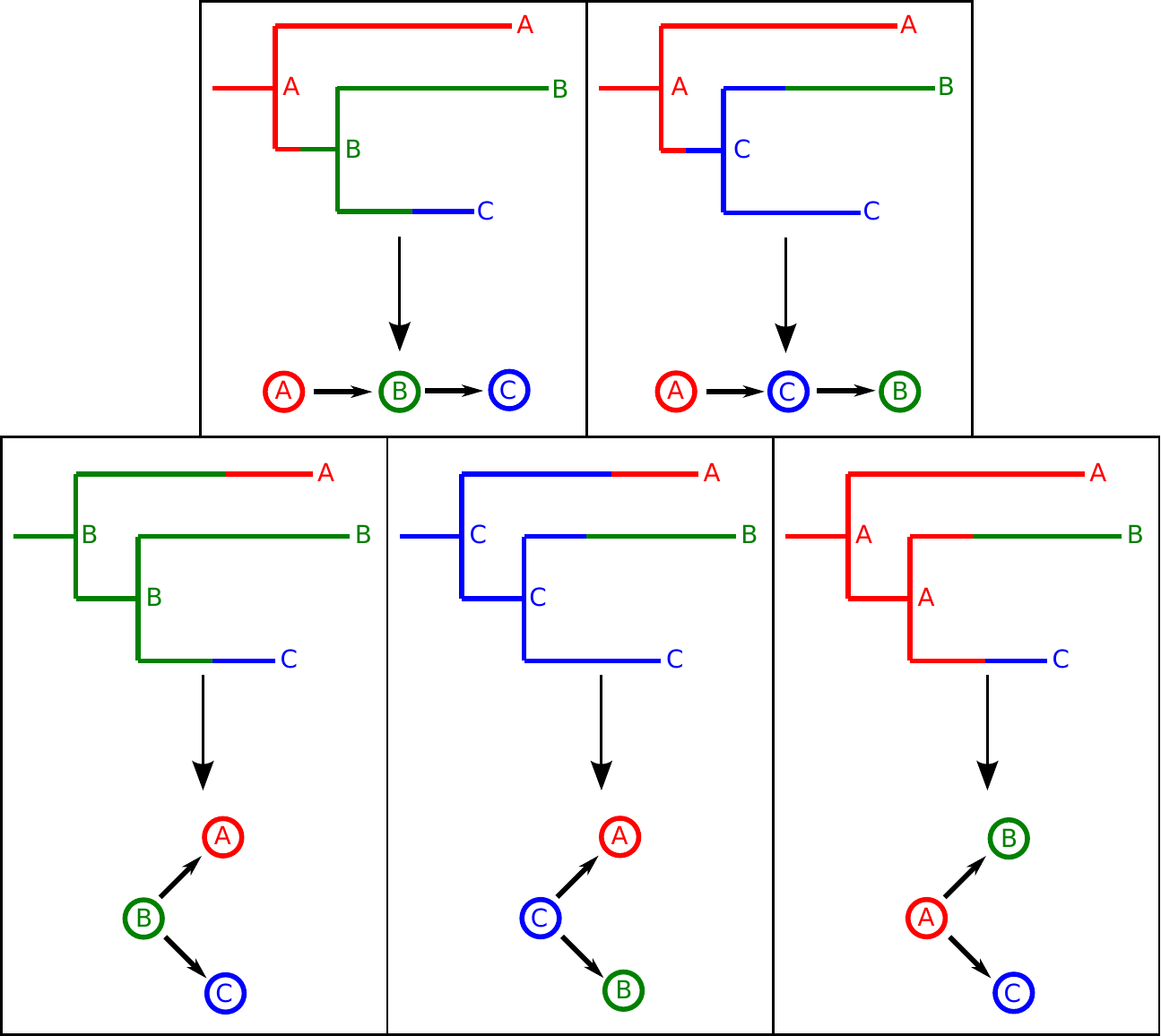}
\caption{The five compatible transmission tree structures of a phylogenetic tree with three tips, depicted as partitions of the phylogeny (above) and as directed graphs amongst the hosts A B and C (below)}\label{sameoldsameold}
\end{figure*}

\begin{figure*}
\centering
\includegraphics[width=16.0cm]{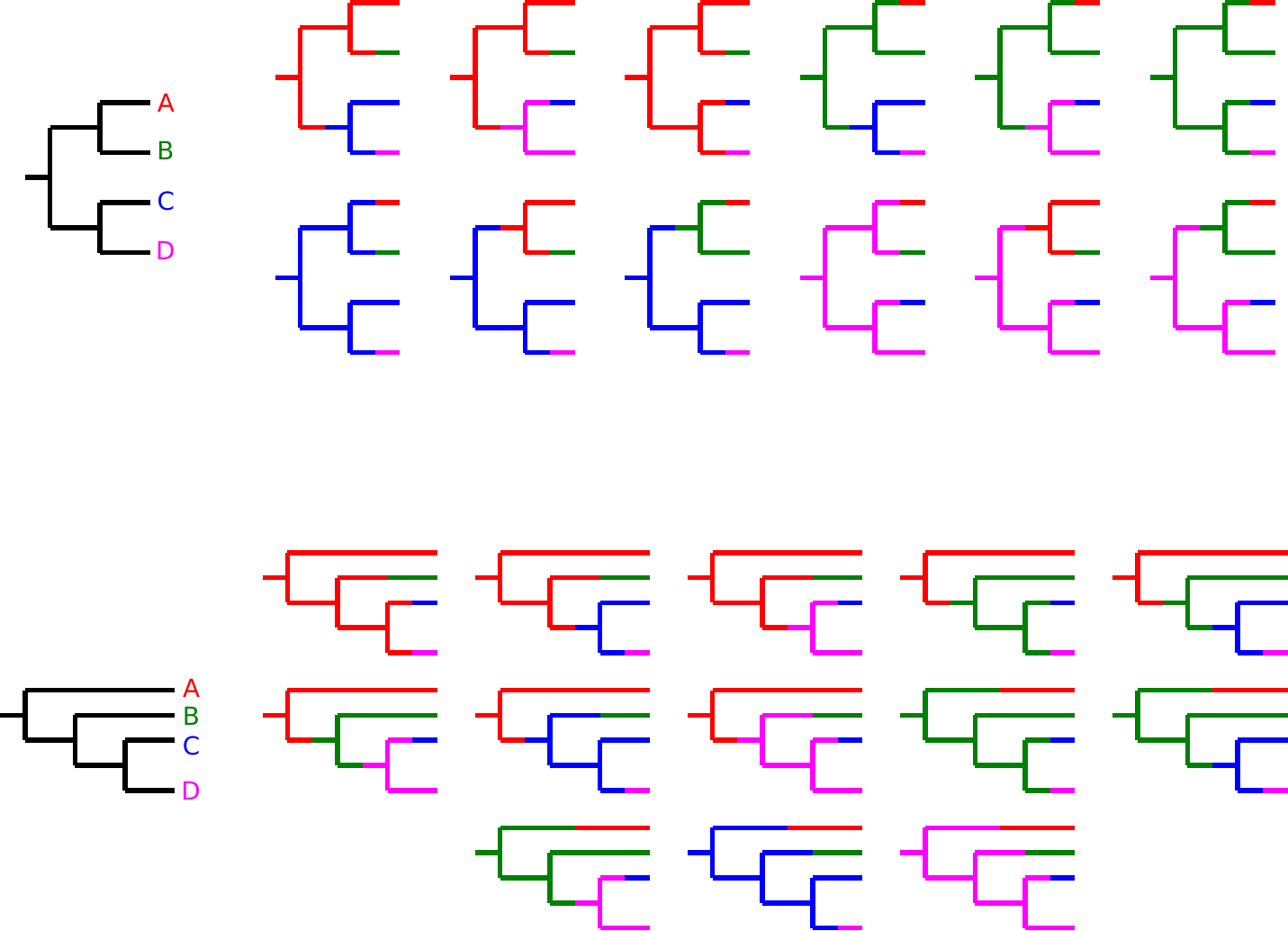}
\caption{Above: the twelve valid partitions of the phylogeny ((A,B),(C,D)). Below: the thirteen valid partitions of the phylogeny (A,(B,(C,D)))}\label{4partitions}
\end{figure*}

\begin{figure*}
\centering
\includegraphics[width=18.0cm]{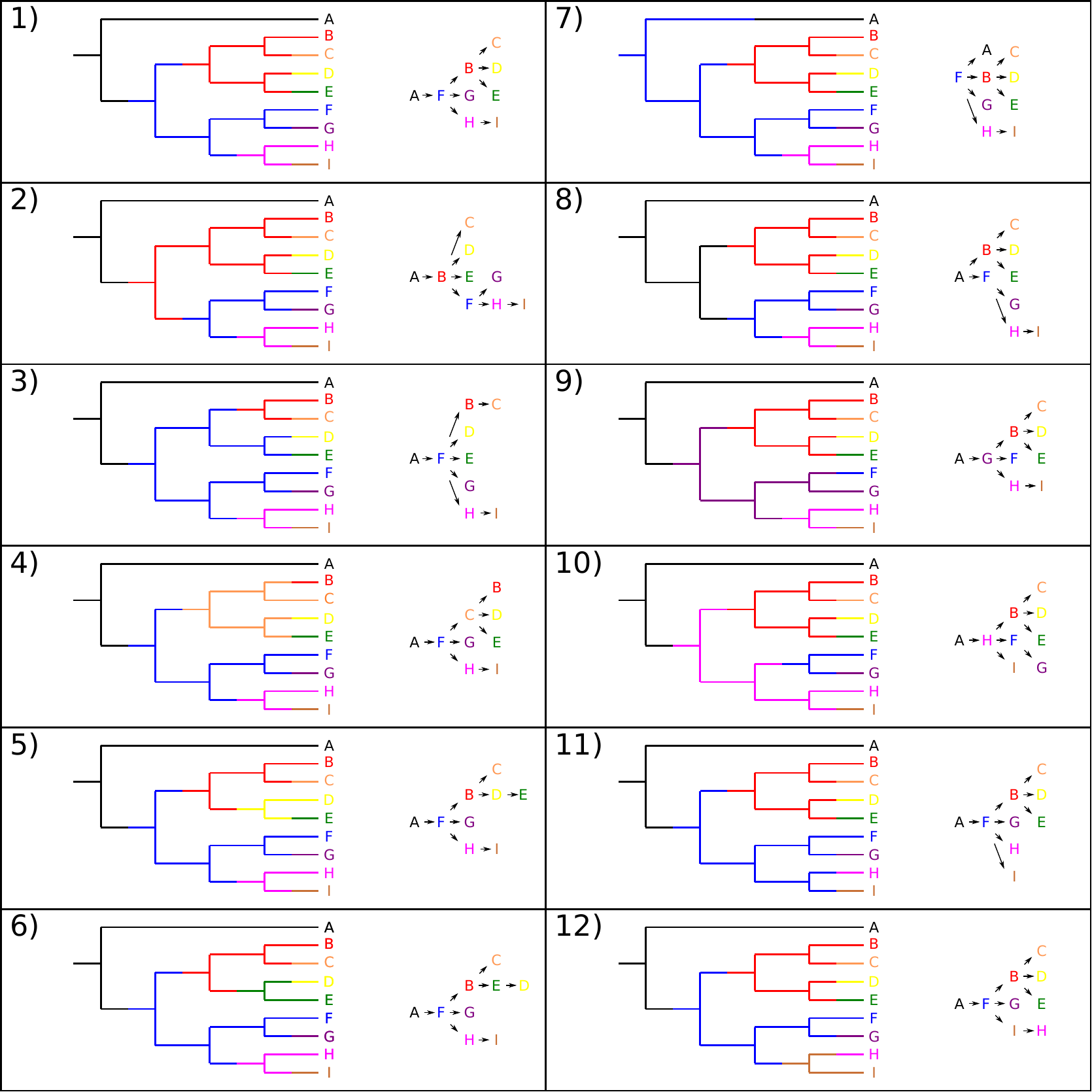}
\caption{Illustration of the effects of the infection branch operator on the partition $\mathcal{P}$ of a phylogeny of samples from the set of hosts A-I, and corresponding effects on the transmission tree. 1. Original partition. 2. Downward move on B. 3. Upward move on B. 4. Downward move on C. 5. Downward move on D. 6. Downward move on E. 7. Downward move on F. 8. Upward move on F. 9. Downward move on G. 10. Downward move on H. 11. Upward move on H. 12. Downward move on I.}\label{ttoperator}
\end{figure*}

\newpage

\bibliography{Partitions_references_short}

\end{document}